\documentclass[onefignum,onetabnum]{siamart190516}
\hypersetup{colorlinks=true,linkcolor=blue,citecolor=blue}

\usepackage{setspace,url}
\usepackage[algo2e,linesnumbered,vlined,ruled]{algorithm2e}
\usepackage{graphics,graphicx,epsf,epstopdf,subfigure}
\usepackage{comment}
\usepackage{times}
\usepackage{amsmath,amsxtra,amsfonts,amscd,amssymb,bm}
\allowdisplaybreaks % automatic page breaks in math environments
\usepackage{supertabular,multirow}
\usepackage{booktabs}
\usepackage[normalem]{ulem}
\usepackage[title]{appendix}
\usepackage{exscale}
\usepackage{hyperref}
\usepackage{tikz,array}
\usetikzlibrary{shapes.geometric}
\usetikzlibrary{shapes.arrows}

\renewcommand{\theenumi}{\roman{enumi}}

\newtheorem{assumption}[theorem]{Assumption}

\newcommand{\crefalg}[1]{\hyperref[#1]{Algorithm~\ref*{#1}}}

\newcommand{\cLb}{{\mathcal{L}_{\beta}}}

\newcommand{\cL}{\mathcal{L}}

\newcommand{\Et}{E_{\mathrm{KS}}}

\newcommand{\R}{\mathbb{R}}
\newcommand{\Rn}{\mathbb{R}^{n}}  
\newcommand{\Rnp}{\mathbb{R}^{n\times p}} 
\newcommand{\Rnn}{\mathbb{R}^{n\times n}}   
\newcommand{\Sn}{\mathbb{SR}^{n\times n}}
\newcommand{\Sp}{\mathbb{SR}^{p\times p}}

\newcommand{\TX}{{\mathcal{T}_{\stiefel^B}}(X)}

\newcommand{\us}{\underline{\sigma}}

\newcommand{\diag}{\mathrm{diag}}

\newcommand{\Diag}{\mathrm{Diag}}

\newcommand{\ff}{_{\mathrm{F}}}
\newcommand{\fs}{^2_{\mathrm{F}}}

\newcommand{\inv}{^{-1}}
\newcommand{\st}{\mathrm{s.\,t.}\,\,} 
\newcommand{\stiefel}{{\cal S}_{n,p}}

\newcommand{\tr}{\mathrm{tr}}
\newcommand{\zz}{^{\top}}

\newcommand{\dkh}[1]{\left(#1\right)}

\newcommand{\norm}[1]{\left\|#1\right\|}
\newcommand{\abs}[1]{\left|#1\right|}
\newcommand{\jkh}[1]{\left\langle#1\right\rangle}

\graphicspath{ {./images/} }
\definecolor{Gray}{rgb}{0.5,0.5,0.5}
\definecolor{myred}{rgb}{0.7,0,0.1}
\definecolor{mygreen}{rgb}{0,0.6,0.2}
\definecolor{myblue}{rgb}{0.5,0,1}
\definecolor{myrvs}{rgb}{0.25,0.45,0.85}

\begin{document}

\title{An orthogonalization-free parallelizable framework for all-electron calculations in density functional theory\thanks{Submitted to the editors July 28, 2020.\funding{BG was supported by the Fonds de la Recherche Scientifique -- FNRS and the Fonds Wetenschappelijk Onderzoek -- Vlaanderen under EOS Project no. 30468160. GH was supported from FDCT of Macao SAR (FDCT 029/2016/A1), MYRG of University of Macau (MYRG2017-00189-FST, MYRG2019-00154-FST), and National Natural Science Foundation of China (Grant Nos. 11922120, 11871489, and 11401608). YK was supported by the Academic Research Fund of the Ministry of Education of Singapore under grant No. R-146-000-291-114. XL was supported in part by the National Natural Science Foundation of China (No. 11971466, 11991021 and 11991020), Key Research Program of Frontier Sciences, Chinese Academy of Sciences (No. ZDBS-LY-7022), the National Center for Mathematics and Interdisciplinary Sciences, Chinese Academy of Sciences and the Youth Innovation Promotion Association, Chinese Academy of Sciences.}}}

\headers{Orthogonalization-free framework for KSDFT}{B. Gao, G. Hu, Y. Kuang, and X. Liu}

\author{
	Bin Gao\thanks{ICTEAM Institute, UCLouvain, Louvain-la-Neuve, Belgium (\email{gaobin@lsec.cc.ac.cn}).}
	\and
	Guanghui Hu\thanks{Department of Mathematics, University of Macau, Macao SAR, China;	Zhuhai UM Science \& Technology Research Institute, Guangdong Province, China (\email{garyhu@umac.mo}).}
	\and
	Yang Kuang\thanks{Corresponding author. Department of Mathematics, National University of Singapore, Singapore (\email{matkuan@nus.edu.sg}).}
	\and
	Xin Liu\thanks{State Key Laboratory of Scientific and Engineering Computing, Academy of Mathematics and Systems Science, Chinese Academy of Sciences, and University of Chinese Academy of Sciences, China (\email{liuxin@lsec.cc.ac.cn}).}
}

\maketitle
\begin{abstract}
All-electron calculations play an important role in density functional theory, in which improving computational efficiency is one of the most needed and challenging tasks. In the model formulations, both nonlinear eigenvalue problem and total energy minimization problem pursue orthogonal solutions. Most  existing algorithms for solving these two models invoke orthogonalization process either explicitly or implicitly in each iteration. Their efficiency suffers from this process in view of its cubic complexity and low parallel scalability in terms of the number of electrons for large scale systems. To break through this bottleneck, we propose an orthogonalization-free algorithm framework based on the total energy minimization problem. It is shown that the desired orthogonality can be gradually achieved without invoking orthogonalization in each iteration. Moreover, this framework fully consists of Basic Linear Algebra Subprograms (BLAS) operations and thus can be naturally parallelized. The global convergence of the proposed algorithm is established. We also present a precondition technique which can dramatically accelerate the convergence of the algorithm. The numerical experiments on  all-electron calculations show the efficiency and high scalability of the proposed algorithm.
\end{abstract}



% -----------------------------------------------------
\section{Introduction}
We aim to find the ground state solution of a molecular system from
 all-electron calculations.  In view of Kohn--Sham density
functional theory (KSDFT) \cite{kohn1965self}, this can be achieved by
solving the lowest $p$ eigenpairs of the Kohn--Sham equation:
\begin{equation}\label{eq:KS}
  \left\{ 
    \begin{array}{lr} 
      \hat{H}\psi_l(\bm{r}) = \varepsilon_l\psi_l(\bm{r}),
      &l=1,2,\dots ,p,\\
      \displaystyle\int_{\mathbb{R}^3} \psi_l\psi_{l'} d\bm{r} = \delta_{ll'},
      &l,l'=1,2,\dots ,p, 
    \end{array} \right. 
\end{equation}
where $\hat{H}$ is the Hamiltonian operator, $\psi_l(\bm{r})$ is the $l$-th wavefunction (eigenfunction), $\varepsilon_l$ refers to the corresponding eigenenergy, $\delta_{ll'}$ is the Kronecker delta function, and $p$ denotes the number of electrons. Alternatively, the ground state solution can be obtained by minimizing the total energy with orthogonality constraints
\cite{payne1992iterative}:
\begin{equation} \label{prob:KS-continuous}
\begin{array}{cl}
\min\limits_{ \varPsi}&\Et(\varPsi) \\
\st &  \jkh{\varPsi,\varPsi}=I_p,
\end{array}
\end{equation}
where $\varPsi = (\psi_{1},\psi_{2},\dots,\psi_{p})$, $\Et$ denotes the Kohn--Sham total energy, $\jkh{\cdot,\cdot}$ stands for the inner product, and $I_p$ denotes the $p\times p$ identity matrix. For notation brevity, we drop the subscript and let $I=I_p$. The detailed expressions of the Hamiltonian operator and the Kohn--Sham total energy are introduced in the next section.

\subsection{Literature review and challenges}
In electronic structure calculations, the pseudopotential approaches have proven to be successful in predicting electrical, magnetic and chemical properties for a wide range of materials \cite{pickett1989pseudopotential}.  However, the pseudopotentials can hardly construct the transition metals accurately \cite{levashov2010pseudopotential} and tend to mispredict the material properties under extreme environment \cite{xiao2010first}. As a result, all-electron calculations which treat the Coulomb external potential exactly are in demand.

One of the most challenging aspects in all-electron calculations is the computational efficiency, which is usually dominated by two factors: the singularities arising from the Coulomb external potential and the orthogonality constraints of the wavefunctions.

To handle the singularities, the numerical discretization is generally required to be well designed in a manner such that it is able to capture the sharp variations of the orbitals and meanwhile describe the results on the regions where the orbitals vary slightly with the least effort. We focus on the finite element discretization (FEM) \cite{tsuchida1996adaptive, suryanarayana2010non, bao2012h, chen2014adaptive} since it has local basis and allows a spatially adaptive resolution. Other discretizations handling the singularities can be found in \cite{ahlrichs1989electronic, cohen2013locally} and references therein.

When the quantum system is large, all-electron calculations turn into  expensive \cite{lin2019numerical}. In particular, to keep the orthogonality of the orbitals becomes the bottleneck in most existing algorithms. The self-consistent field (SCF) method and its variants \cite{kohn1965self, kerker1981efficient} are commonly used to solve the KS equation \eqref{eq:KS}. However, the global convergence of the SCF-based algorithms can hardly be guaranteed \cite{liu2014convergence, liu2015analysis} and hence good initial guesses are often crucial for their performance. Since they lack robustness, the bad performance is often observed in numerical experiments \cite{yang2007trust}. This motivates the research on solving the total energy minimization problem \eqref{prob:KS-continuous} directly. Most of the first-order methods, such as QR retraction \cite{zhang2014gradient} and multipliers correction framework \cite{gao2018new}, carry out a feasible update. Namely, certain orthogonalization process is invoked in each iteration. Note that the orthogonalization process costs at least $\mathcal{O}(p^3)$ per  iteration. Hence, these methods are not competent in solving large quantum systems due to this cubic complexity and the low   scalability of any orthogonalization process.

Several algorithms have been exploited to avoid the {orthogonalization}.  Linear scaling methods \cite{bowler2012methods} build the solutions by direct minimization of unconstrained variational formulations. Note that most of them require to estimate the upper bound of the eigenvalue of the Hamiltonian \cite{lin2019numerical}, which is intractable in practice. Recently, an infeasible optimization algorithm based on the augmented Lagrangian method {has been} proposed in \cite{gao2018parallelizable}. Here, ``infeasible" indicates that the iterate is not required to satisfy the constraints in each iteration. The orthogonality can be guaranteed at any cluster point of the iteration sequence. Another favored property of this algorithm is that it is not sensitive with the choices of initial guess and parameters which makes it robust. Moreover, it is illustrated both theoretically and numerically that this algorithm does not highly rely on any priori knowledge of the studied system.  All the calculations in which fully consist of BLAS operations. Thus it can be naturally parallelized. In view of these features, a parallelizable framework based on this infeasible minimization method for all-electron calculations is proposed.

\subsection{Contribution}
In this paper, we provide a competitive algorithm framework for all-electron calculations in the density functional theory. The framework consists of four parts shown in \cref{fi:frameKS}, i.e., the pre-processing part for configuring the problem, the discretization part for numerically discretizing the continuous problem, the solving part for obtaining the solutions of the discretized system, and the post-processing part for transforming the numerical solutions for the further applications.

\begin{figure}[htbp]
	\centering \footnotesize
	\begin{tikzpicture}[
	scale=.85,
	auto,
	decision/.style = { diamond, aspect=2, draw=gray,
		thick, fill=gray!10, text width=4em, text badly centered,
		inner sep=1pt},
	block/.style = { rectangle, draw=gray, thick, fill=gray!10,
		text width=6em, text centered, rounded corners,
		minimum height=2em },
	line/.style = { draw, thick, ->, shorten >= 0.5pt},
	]
	
	\node [block] at (0,0) (prob) {pre-processing};
	
	\node at (3,1.7) (null1) {};
	
	\node [block] at (2.75,0) (genh) {discretization};
	
	\node [block,rectangle, draw=gray, thick,dashed, fill=gray!0,
	text width=13.5em, text centered, rounded corners, minimum
	height=6.5em] at (6.5,0.35) (solver) {};
	
	\node [color=gray] at (6.5,-1.05) {solving};
	
	\node [block,text width=4em,minimum height=3em] at (5.2,0)
	(iteration) {main iteration};
	
	\node [decision] at (7.5,0) (decision) {stop?};
	
	\node [block,text width=7em] at (10.5,0) (end0) {post-processing};
		
	\node [above] at (7.2,0.5) {No};
	
	\begin{scope} [every path/.style=line,thick,shorten >= 0.5pt]
	\path (prob)    --   (genh);   
	\path (genh)      --   (iteration);
	\path (iteration)      --   (decision);
	\path (decision) --   node {Yes} (end0); 
	\path (decision) -- (7.5,1.1) -- (5.2,1.1) -- (iteration);
	\end{scope}
	
	\end{tikzpicture}
	\caption{Flowchart of the framework for ground state
		calculations. \label{fi:frameKS}}
\end{figure}
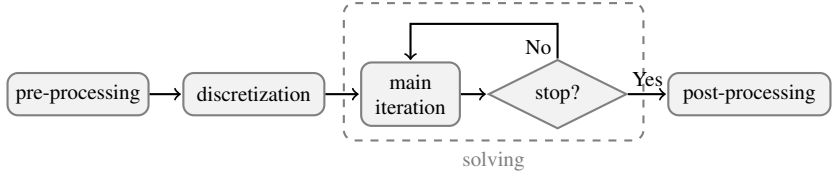

The efficiency of all-electron calculations benefits from the following aspects of the proposed framework in \cref{fi:frameKS}: i). a quality finite element space is designed for the given electronic structure based on the \emph{a priori} analysis; ii). an orthogonalization-free method is proposed and analyzed for the discretized minimization problem; iii). high scalability is successfully demonstrated by numerical examples.

More specifically, in preparation of the tetrahedron mesh, the decay of the external potential is studied with the linear interpolation theory in \cite{huang2010adaptive, suryanarayana2010non, kuang2019on}, and a strategy on generating radial mesh for optimally capturing such decay is designed for a given electronic structure. It is noted that a quality finite element space would be built based on the radial mesh, and the efficiency of the algorithm would benefit from the sparsity of the discretized system and the mature and robust solvers for the sparse system such as the algebraic multigrid method.

The new method for the discretized optimization problem \eqref{prob:KS-continuous} is proposed by extending the parallelizable column-wise augmented Lagrangian  (PCAL) \cite{gao2018parallelizable} from the following two aspects. First, the PCAL is revised to handle the minimization problem with general orthogonality constraints $X^{\top}BX = I$ rather than the standard ones $X^{\top}X = I$. The global convergence of the new method is established theoretically. Second, a precondition strategy is proposed for the class of the PCAL methods, and a specific preconditioner is designed for all-electron calculations, which brings the dramatic acceleration for the convergence in the simulations.

As an attractive feature of the proposed algorithm, the robustness is successfully shown by a variety of numerical experiments, i.e., a random initial guess works for all numerical experiments in this paper, and the numerical convergence of the algorithm is not sensitive to the selection of the parameters. Finally, the high scalability of the algorithm is demonstrated by the numerical examples, which obviously indicates the potential of our algorithm for the large scale systems.

\subsection{Notation and organization}
$\Sp:=\{S\in\R^{p\times p}\mid S\zz = S\}$ refers to the set of $p\times p$ real symmetric matrices. $\sigma_{\min}(A)$ denotes the smallest singular value of given real matrix $A$. $\Diag(v)\in\Sn$ denotes a diagonal matrix with all entries of $v\in\Rn$ in its diagonal, and $\diag(A)\in\Rn$ extracts the diagonal entries of matrix $A\in\Rnn$. For convenience, $\Theta(M):=\Diag(\diag(M))$ represents the diagonal matrix with the diagonal entries of square matrix $M$ in its diagonal. $\mathrm{sym}(A):=\frac{1}{2}(A+A\zz)$ stands for the average of a square matrix and its transpose.

The organization of this paper is as follows. The KSDFT and its discretization are presented in \cref{sec:ksdft}. In \cref{sec:alg}, we present the algorithm and its convergence results. The implementation details of the proposed framework are introduced in \cref{sec:imp} and the numerical experiments are reported in \cref{sec:num}. In the end, we draw a brief conclusion and introduce the future works.

\section{Finite Element Discretization for KSDFT} \label{sec:ksdft}

In this section, we introduce the detailed formulations for KSDFT and
the discretization part as illustrated in \cref{fi:frameKS}.

\subsection{KSDFT}
We consider a molecular system in $\mathbb{R}^3$ consisting of $M$ nuclei of charges ${Z_1,\dots,Z_M}$ locating at the positions ${\bm{R}_1,\dots,\bm{R}_M}$ and $p$ electrons in the non-relativistic setting. The atomic unit is adopted in this work. Thus the Hamiltonian operator $\hat{H}$ in the Kohn--Sham equation \eqref{eq:KS} can be written as
\begin{equation}\label{eq:Hamiltonian}
  \hat{H}  = -\frac{1}{2}\nabla^2 +   V_{\mathrm{ext}}(\bm{r})
  +V_{\mathrm{Har}}([\rho];\bm{r})+ V_{\mathrm{xc}}([\rho];\bm{r}),
\end{equation}
where the notation $V([\rho];\bm{r})$ implies that $V$ is a functional of the electron density
$ \rho(\bm{r})=\sum_{l=1}^{p}\lvert \psi_l(\bm{r})\lvert^2$. The first term  $-\nabla^2/2$ in $\hat{H}$ is the kinetic operator. The second term in $\hat{H}$ describes the Coulomb external potential due to the nuclei which takes the form
\begin{equation}\label{eq:ext}
V_{\mathrm{ext}}(\bm{r})=-\sum_{j=1}^{M}\frac{Z_j}{\lvert
	\bm{r}-\bm{R}_j\rvert}.
\end{equation}
The third term is the Hartree potential describing the Coulomb repulsion among the electrons
\begin{equation}\label{eq:har}
  V_{\mathrm{Har}}([\rho];\bm{r})=\int_{\mathbb{R}^{3}} \frac{\rho(\bm{r'})}
  {\lvert \bm{r}-\bm{r'} \lvert} d\bm{r'}.
\end{equation}
The last term $V_{\mathrm{xc}}$ stands for the exchange-correlation potential, which is caused by the Pauli exclusion principle and other non-classical Coulomb interactions. Note that the analytical expression for the exchange-correlation term is unknown and therefore an approximation is needed. Specifically, the local density approximation (LDA) from the library Libxc \cite{marques2012libxc} is adopted in this work.

The total energy of the given quantum system consists of several
parts:
\begin{equation}\label{eq:totalEnergy}
  \Et = E_{\mathrm{kinetic}}+ E_{\mathrm{ext}}
  + E_{\mathrm{Har}} 
  + E_{\mathrm{xc}} +  E_{\mathrm{nuc}},
\end{equation}
where $E_{\mathrm{kinetic}}$ is the kinetic energy, and
$E_{\mathrm{ext}}$, $E_{\mathrm{Har}}$, $E_{\mathrm{xc}}$, and
$E_{\mathrm{nuc}}$ are the potential energies induced by
$V_{\mathrm{ext}}$, $V_{\mathrm{Har}}$, $V_{\mathrm{xc}}$, and the
nucleus-nucleus potential, respectively. 
Denoting the exchange-correlation energy per particle by
$\epsilon_\mathrm{xc}(\rho)$, then $V_\mathrm{xc}$ is the functional derivative of
$\epsilon_\mathrm{xc}(\rho)$ with respect to $\rho$, i.e.,
$V_{xc} = \delta\epsilon_\mathrm{xc}(\rho)/\delta\rho$. As a result, it follows that
\begin{align*}\label{eq:energies}
  E_{\mathrm{kinetic}} &= \frac{1}{2}\sum_{l=1}^{p}\int_{\mathbb{R}^3}\lvert \nabla
                         \psi_l\lvert^2 d\bm{r}, \quad
  E_{\mathrm{ext}} = \int_{\mathbb{R}^3}V_{\mathrm{ext}} \rho(\bm{r}) d\bm{r},\quad
  E_{\mathrm{Har}} = \frac{1}{2} \int_{\mathbb{R}^3} V_{\mathrm{Har}}\rho(\bm{r})d\bm{r}, \\ E_{\mathrm{xc}} &= \int_{\mathbb{R}^3} \epsilon_{\mathrm{xc}}\rho(\bm{r})d\bm{r},\quad
  E_{\mathrm{nuc}} = \sum_{j=1}^M\sum_{k=j+1}^{M}\frac{Z_jZ_k}{\left|
                    \bm{R}_j-\bm{R}_k \right|}.
\end{align*}
Note that $E_{\mathrm{nuc}}$ is a constant for the given system.

The ground state of the given system can be obtained from solving either the KS equation \eqref{eq:KS} or the total energy minimization problem \eqref{prob:KS-continuous}. In order to numerically solve the continuous problem, we consider the finite element discretization.

\subsection{Finite element discretization} \label{sec:fem} 
In practical simulations, a bounded polyhedral domain
$\Omega \subset \mathbb{R}^3$ is served as the computational
domain. Thus the variational form of the Kohn--Sham equation
\eqref{eq:KS} on $\Omega$ can be formulated as: Find
$(\varepsilon_l,\psi_l)\in \mathbb{R}\times H_0^1(\Omega)$,
$l=1,2, \dots, p$, such that
\begin{equation}
  \label{eq:vf}
  \left\{
  \begin{array}{l}
    \displaystyle\int_{\Omega} \varphi \hat{H} \psi_l d\bm{r} =
    \varepsilon_l\int_{\Omega}\psi_l \varphi d\bm{r}, 
    \quad \forall \varphi\in H_0^1(\Omega),\\[3mm]
    \displaystyle\int_{\Omega}\psi_l\psi_{l'} d\bm{r} = \delta_{ll'},
    \quad l' = 1,2,\dots, p,
  \end{array}\right.
\end{equation}  
where
$H_0^1(\Omega) = \{\varphi \in H^1(\Omega):\varphi|_\Omega=0\}$ and $H^1(\Omega)$ is a standard Sobolev space. 

To build a high quality finite element space to approximate the solution of \eqref{eq:vf} in all-electron calculations, the singularities stemming from the Coulomb potential in \eqref{eq:ext} should be prudently treated. In this work, we adopt a radial mesh generation strategy to resolve the difficulty brought by the singularities; see \cref{sec:radialmesh} for details.%There are several strategies to resolve the difficulty brought by the singularities. In this work, we adopt a radial mesh generation strategy; see \cref{sec:radialmesh} for details.

Assume that the linear finite element space $V_h\subset H_0^1(\Omega)$ is constructed on the bounded domain $\Omega$ partitioned by
$\mathcal{T} =\{\mathcal{T}_K,K=1,2,\dots,N_\mathrm{ele}\}$, where
$N_\mathrm{ele}$ represents the total number of elements of
$\mathcal{T}$. Several commonly used notations in $V_h$ are defined
here. The basis functions are denoted by $\varphi_i$, $i=1,\dots ,n$,
where $n$ is the dimension of $V_h$ and the set of basis functions is
denoted by $ \mathcal{N} = (\varphi_1,\dots,\varphi_{n})^{\top}$. We
construct the matrix of basis function $\mathcal{B}$ with
$\mathcal{B}_{i,j} = \varphi_i \varphi_j$, then the symmetric mass
matrix $B\in \Sn$ can be obtained from
$B_{i,j}=\int_{\Omega} \mathcal{B}_{i,j}d\bm{r}$. Furthermore, a sequence of matrices
$\{ G_{(l)}\in \Sn, l = 1,\dots,n\}$ with the entries
$( G_{(l)})_{i,j} = \int_{\Omega} \mathcal{B}_{i,j} \varphi_l d\bm{r}$ are
introduced. The discretized Laplacian $L\in \Sn$ on $V_h$ is defined
as
$L_{i,j}=\int_{\Omega}\nabla \varphi_j\cdot \nabla \varphi_i d\bm{r}$.

On the finite element space $V_h$, the discretized variation form of
(\ref{eq:vf}) turns out: Find
$(\varepsilon_l^h,\psi_l^h)\in \mathbb{R}\times V_h$, $l=1,2,\dots,p$,
such that
\begin{equation}
  \label{eq:discretevf}
  \left\{
  \begin{array}{l}
    \displaystyle\int_{\Omega} \varphi \hat{H} \psi_l^h d\bm{r} =
    \varepsilon_l\int_{\Omega}\psi_l^h \varphi d\bm{r}, 
    \quad \forall \varphi\in V_h ,\\[3mm]
    \displaystyle\int_{\Omega}\psi_l^h\psi_{l'}^h d\bm{r} = \delta_{ll'},
    \quad l' = 1,2,\dots, p.
  \end{array}\right.
\end{equation}
We express the $l$-th wavefunction as
$\psi_l^h=\sum_{i=1}^{n}X_{i,l}\varphi_i = X_l^{\top}\mathcal{N}$,
where $X \in \mathbb{R}^{n\times p}$ and $X_{i,l}$~stands for the
$i$-th degree of freedom of $\psi_l^h$. Then the electron density
 can be rewritten~as
\begin{equation*}
  \rho(\bm{r}) =  \sum_{l=1}^p(X_l^{\top}\mathcal{N})(X_l^{\top}\mathcal{N})
  = \sum_{l=1}^{p}X_l^{\top}\mathcal{B}X_l 
  = \mbox{tr}(X^{\top}\mathcal{B}X).
\end{equation*}

Note that the Hartree potential $V_{\mathrm{Har}}$ in \cref{eq:har} is also the solution to the Poisson equation $-\nabla^2 V_{\mathrm{Har}} = 4 \pi \rho(\bm{r})$. We denote the discretized Hartree potential by $U(X)\in\mathbb{R}^n$ such that $V_{\mathrm{Har}}= U(X)\zz\mathcal{N}$. After the finite element discretization on the Poisson equation, $U$~is calculated from the linear system $LU(X) = 4\pi \left(\mbox{tr}(X^{\top}G_{(1)}X),\dots,\mbox{tr}(X^{\top}G_{(n)}X)\right)^{\top}$. In practical simulations, this linear system is solved by an efficient algebraic multigrid method~\cite{bao2012h}. 
%As a result, $U(X) = 4\pi L^{-1} \left(\mbox{tr}(X^{\top}G_1X), \dots, \mbox{tr}(X^{\top}G_nX)\right)^{\top}$.

Due to the arbitrary of $\varphi$ in \eqref{eq:discretevf}, we can
choose $\varphi = \varphi_i, i=1,\dots,n$. In view
of above expressions, finding the solution of the discretized variational form
\eqref{eq:discretevf} turns out solving the generalized nonlinear
eigenvalue problem:
\begin{equation}\label{eq:KS-gep}
  \left\{
    \begin{array}{ll}
      H(X)X= BX\Xi, \\ [6pt]
      X^{\top} BX = I_{p},
    \end{array}
  \right.
\end{equation}
where $\Xi = \Diag(\varepsilon_1^h,\dots,\varepsilon_p^h)$, $H(X)\in\Sn$ is the discretized Hamiltonian matrix which can be formulated from \eqref{eq:Hamiltonian} as
\begin{equation}\label{eq:hamiltonian} 
  H(X) = \frac{1}{2}L + M_{\mathrm{ext}}+ M_{\mathrm{Har}}(X)+ M_{\mathrm{xc}}(X).
\end{equation} 
The matrices
$M_{\mathrm{ext}}, M_{\mathrm{Har}}(X), M_{\mathrm{xc}}(X)\in
\mathbb{R}^{n\times n}$ are defined as
\begin{equation*}
  (M_{\mathrm{ext}})_{i,j}=\int_{\Omega}V_{\mathrm{ext}}\mathcal{B}_{i,j} d\bm{r},\quad
  (M_{\mathrm{Har}})_{i,j}=\int_{\Omega}V_{\mathrm{Har}}\mathcal{B}_{i,j} d\bm{r},\quad
  (M_{\mathrm{xc}})_{i,j}=\int_{\Omega}V_{\mathrm{xc}}\mathcal{B}_{i,j} d\bm{r}.
\end{equation*}

We now represent the total energy \eqref{eq:totalEnergy} in the
discretized form:
\begin{align*}
  E_\mathrm{kinetic}(X) &= \frac{1}{2}\sum_{l=1}^p\int_{\Omega}\nabla
  \psi_l\cdot \nabla \psi_l d\bm{r} = \frac{1}{2}\sum_{l=1}^p
  \int_{\Omega}X_l^{\top}\nabla\mathcal{N} \cdot X_l^{\top}\nabla \mathcal{N} d\bm{r}
  = \frac{1}{2} \mbox{tr}(X^{\top}LX),\\
  E_\mathrm{ext}(X) &=  \int_{\Omega}V_{\mathrm{ext}}\rho(\bm{r}) d\bm{r}
  = \int_{\Omega} V_{\mathrm{ext}}\mbox{tr}(X^{\top}\mathcal{B}X) d\bm{r} = \mbox{tr}(X^{\top}M_{\mathrm{ext}}X),  \\  
   E_\mathrm{Har}(X) &= \frac{1}{2}\int_{\Omega} V_{\mathrm{Har}}\rho(\bm{r}) d\bm{r}
  = \frac{1}{2}\int_{\Omega} V_{\mathrm{Har}}\mbox{tr}(X^{\top}\mathcal{B}X) d\bm{r}  = \frac{1}{2}\mbox{tr}(X^{\top}M_{\mathrm{Har}}(X)X), \\  
  E_\mathrm{xc}(X) &= \int_{\Omega} \epsilon_{\mathrm{xc}} \rho(\bm{r}) d\bm{r} =  \int_{\Omega} \epsilon_{\mathrm{xc}}\mbox{tr}(X^{\top}\mathcal{B}X) d\bm{r} = \mbox{tr}(X^{\top}M_{\mathrm{exc}}(X)X),
\end{align*}
where the matrix $M_{\mathrm{exc}}(X)$ in the last formula is defined
as
$(M_{\mathrm{exc}})_{i,j}=\int_{\Omega}\varepsilon_{\mathrm{xc}}\mathcal{B}_{i,j}
d\bm{r}$. Thus the discretized form of the minimization problem
\eqref{prob:KS-continuous} is assembled as
\begin{equation}\label{prob:KS}
  \begin{array}{cl}
    \min\limits_{X\in\Rnp}&E_{\mathrm{KS}}(X) = E_{\mathrm{kinetic}}(X)+ E_{\mathrm{ext}}(X)
    + E_{\mathrm{Har}}(X) 
    + E_{\mathrm{xc}}(X) +  E_{\mathrm{nuc}}
                            \\
    \st &  X^{\top}BX = I_p.
  \end{array}
\end{equation}
The generalized orthogonality constraints in \eqref{prob:KS} are known as the generalized Stiefel manifold \cite{absil2009optimization}, denoted by $\stiefel^B:=\{X\in \Rnp: X\zz B X = I_p\}$. Note that the gradient of $\Et(X)$ satisfies $\nabla \Et(X)=2H(X)X$, while we scale it as $\nabla \Et(X) = H(X)X$ to be consistent with the convention.

\section{Parallelizable Algorithms}\label{sec:alg}
In this section, we concentrate on the solving part in \cref{fi:frameKS}. Namely, the discretized total energy minimization problem \eqref{prob:KS} is considered. We first state its optimality condition. Then a one step gradient-descent update is proposed for solving \eqref{prob:KS} and its global convergence result is established. We also develop an upgraded algorithm based on the column-wise block minimization with preconditioning. 

The discretized total energy minimization problem \eqref{prob:KS} is a nonconvex constrained optimization problem due to the orthogonality constraints. We state its first-order optimality condition as follows.
\begin{definition}\label{def:1}
	Given $X\in\Rnp$, we call $X$ a first-order stationary point of \eqref{prob:KS} if the following condition
	\begin{equation} \label{eq:FONS}
	\left\{
	\begin{array}{ccc}
	\tr(Z\zz \nabla \Et(X))&\geq& 0,\\
	X\zz BX &=& I_p
	\end{array}
	\right.
	\end{equation}
	holds for any $Z\in\TX$, where $\TX:=\{Z\in\Rnp: Z\zz BX+X\zz BZ=0 \}$ is the tangent space of $\stiefel^B$ at $X$.
\end{definition}

Following from \cite[Lemma 2.2]{gao2018new}, it can be proved that the condition \eqref{eq:FONS} is equivalent~to
\begin{equation} \label{eq:FON}
\left\{
\begin{array}{ccc}
(I_n-BXX\zz) \nabla \Et(X) &= & 0, \\ %& \qquad{\mbox{\bf sub-stationarity}}\\
X\zz \nabla \Et(X) & = & \nabla \Et(X)\zz X, \\ %& \qquad{\mbox{\bf symmetry}}\\
X\zz BX &=& I_p. % & \qquad{\mbox{\bf feasibility}}
\end{array}
\right.
\end{equation}
In fact, the second equality of \eqref{eq:FON} is automatically satisfied since $\nabla \Et(X)=H(X)X$ and the Hamiltonian $H(X)$ is symmetric.   Moreover, the condition \eqref{eq:FON} can be further reformulated as
\begin{equation}\label{eq:kkt}
\left\{
\begin{array}{l}
\nabla \Et(X) = BX\Lambda,\\ %\,\, \mbox{with~} \Lambda =\nabla \Et(X)\zz X;
X\zz BX = I_p,
\end{array}
\right.
\end{equation}
where the symmetric matrix $\Lambda \in \Sp$ can be regarded as
the Lagrangian multipliers of the generalized orthogonality
constraints. Multiplying the first equation from the left by $X\zz $,
it follows that $\Lambda$ reads the closed-form expression at any
first-order stationary point,
\begin{equation}\label{eq:closed-multi}
\Lambda = X\zz\nabla \Et(X) = X\zz H(X) X.
\end{equation}

\subsection{Main iteration: one step gradient-descent update}\label{subsec:MALM}
The infeasible method proposed in \cite{gao2018parallelizable} has
been proven to be efficient for solving the lager scale orthogonality constrained optimization problems. Briefly, the iterates are not required to be orthogonal. Meanwhile, the feasibility violation gradually decreases to zero until the method converges. This type of methods enables us 
to get rid of the unscalable computation for preserving constraints. 
In addition, it provides an opportunity to employ the multi-core machines and thus gain more scalability from parallel computing.

The algorithm in \cite{gao2018parallelizable} originally aims to solve the problem with orthogonality constraints ($X\zz X=I$), and
in this subsection, we extend it to the general case ($X\zz BX=I$) which is
not a trivial task.
The skeleton of this algorithm is based on the augmented Lagrangian method (ALM) \cite{nocedal2006numerical}. Let $X^k$ be the current iterate, the classical ALM has two major steps in each iteration:
\begin{itemize}
	\item[1)]  Update the Lagrangian multipliers $\Lambda^k$;
	\item[2)] Minimize the ALM subproblem to obtain $X^{k+1}$, 
	\begin{equation}\label{prob:ALM-sub}
	\min\limits_{X\in\Rnp} {\cL}_\beta(X,\Lambda^k):=\Et(X) - \frac{1}{2}\jkh{\Lambda^k, 
		X\zz BX-I_p} + \frac{\beta}{4} \norm{X\zz BX-I_p}\fs,
	\end{equation}
	where $\cLb(X,\Lambda^k)$ defines the augmented Lagrangian function of problem \eqref{prob:KS} and $\beta>0$ is the penalty parameter.
\end{itemize}  
This framework avoids being confronted with the generalized orthogonality constraints. Next, we discuss how to update these two steps efficiently. 

For step 1), in view of the fact \eqref{eq:closed-multi}, we suggest the following update of Lagrangian multipliers
\begin{equation}\label{eq:multi}
\Lambda^k = {X^k}\zz H(X^k) X^k.
\end{equation}
Due to the symmetry of the Hamiltonian $H(X^k)$, the above update provides symmetric multipliers $\Lambda^k$, which allows us to waive the symmetrization step,  $\mathrm{sym}({X^k}\zz H(X^k) X^k)$, in~\cite{gao2018parallelizable}.

On the other side, the ALM subproblem in step 2) is an unconstrained optimization problem, and various methods can be applied to derive different updates. Instead of solving the subproblem to a certain preset precision, our strategy is to provide an approximate solution by an explicit formulation. We first introduce a proximal linearized approximation~\cite{bolte2014proximal} to substitute the augmented Lagrangian function in \eqref{prob:ALM-sub}. Specifically, we consider the~subproblem
\begin{equation}\label{eq:PLLag}
\min\limits_{X\in\Rnp} \, %\tilde{\cL}_\beta(X)=
\jkh{\nabla_X \cLb(X^k,\Lambda^k), X-X^k} %\tr(\nabla_X \cLb(X^k,\Lambda^k)\zz (X-X^k))
+ \frac{\eta_k}{2}
\norm{X-X^k}\fs.
\end{equation}
The parameter $\eta_k$ measures the dominance of the proximal term. The solution of this quadratic subproblem reads an explicit form 
\begin{align} 
X^{k+1} &= X^k - \frac{1}{\eta_k}\nabla_X \cLb(X^k,\Lambda^k)\nonumber\\
&= X^k - \frac{1}{\eta_k} \dkh{H(X^k)X^k -  BX^k {X^k}\zz H(X^k) X^k
	+ \beta BX^k({X^k}\zz BX^k - I_p)}, \label{eq:PLAM-main}
\end{align}
where the last step is owing to the update formula \eqref{eq:multi}. It implies that this modified ALM update is nothing but a vanilla gradient-descent step and $1/\eta_k$ specifies the stepsize.  

Now we turn back to the solving part in \cref{fi:frameKS}. By using the one step gradient-descent update \eqref{eq:PLAM-main} in the main iteration, we fulfill a solving part for KSDFT. The complete algorithm is described in \crefalg{alg:PLAM}. 

\begin{algorithm2e}[ht]
	\caption{Proximal Linearized Augmented Lagrangian Algorithm (PLAM)}
	\label{alg:PLAM}
	\SetKwInOut{Input}{input}\SetKwInOut{Output}{output}
	\SetKwComment{Comment}{}{}
	%\lines numbered
	\BlankLine %\dontprintsemicolon
	\textbf{Input:} discretization with $n\in\mathbb{N}$ and $B\in\mathbb{SR}^{n\times n}$; tolerance $\epsilon>0$; initial guess $X^0\in\Rnp$; Set $k:=0$.\\
	\While{$\norm{(I_n-BX^k{X^k}\zz)H(X^k)X^k}\ff + \norm{{X^k}\zz BX^k-I}\ff > \epsilon$\label{al:stepW}}
	{
		Compute the  Hamiltonian $H(X^k)$ by \eqref{eq:hamiltonian}.\label{al:stepH}
		
		Update the variable $X^{k+1}$ by \eqref{eq:PLAM-main}.\label{al:stepX}
		%		$$X^{k+1} = X^k - \frac{1}{\eta_k}\left( H(X^k)X^k -  BX^k {X^k}\zz H(X^k) X^k + \beta BX^k({X^k}\zz BX^k - I_p)\right).$$
		
		Update the parameters $\eta_k$ and $\beta$; Set $k:=k+1$. \label{al:stepP}
	
	}
	
	%	\textbf{Post-processing:} $\bar{X}=\mathrm{orth}(X^k)$.
	
	\textbf{Output:} $X^k$.
\end{algorithm2e}

Once the pre-processing and discretization are finished, the number of degrees of freedom $n$ and the matrix $B$ are fixed. Meanwhile, the initial guess $X^0$ can be generated by any popular strategy in KSDFT. In view of the condition~\eqref{eq:FON}, we notice that Line \ref{al:stepW} (the stopping criteria) in \crefalg{alg:PLAM} is sufficient to check the first-order optimality. Line \ref{al:stepH}-\ref{al:stepP} are the main iterations in \cref{fi:frameKS}. Indeed, those calculations in KSDFT can be well assembled in a parallel way. The gradient-descent update in Line \ref{al:stepX} is the BLAS3 operation. The choices of parameters will be discussed in \cref{subsec:parameters}. To sum up, the algorithm PLAM can be conveniently  implemented since there is no matrix decomposition or eigen-solver. It completely consists of BLAS operations. Therefore, the algorithm PLAM is open to be parallelized. Note that SCF method can also be described by the framework \cref{fi:frameKS}, and the only distinction between SCF and PLAM is the main iteration. Specifically, SCF replaces Line~\ref{al:stepH}-\ref{al:stepP} with solving a linear eigenvalue problem from \eqref{eq:KS-gep}. By contrast, PLAM just carries out a one step gradient-descent update.

\subsection{Convergence analysis}
The global convergence of the plain PLAM for orthogonality constraints ($X\zz X=I$) has been studied in~\cite{gao2018parallelizable}. Next, we consider the generalized case, i.e., $X\zz BX=I$. It can be proved that the existing results are still applicable for \crefalg{alg:PLAM}. 

A natural idea to investigate the generalized orthogonality constraints is transforming it into the standard case. Since $B$ is symmetric positive definite, there exists a symmetric positive definite matrix $G\in\Rnn$
satisfying $B=G^2$. By taking $Y=GX$, the problem~\eqref{prob:KS} is equivalent to
\begin{equation}\label{prob:KS-transform-Y}
\begin{array}{cc}
\min\limits_{Y\in\Rnp}&g(Y):=\Et(G\inv Y)\\
\mbox{s.t.} & Y\zz Y=I_p.
\end{array}
\end{equation}
Thus the augmented Lagrangian function of \eqref{prob:KS-transform-Y} is defined as
\begin{equation*}\label{eq:Lag-Y}
\tilde{\cLb}(Y,\tilde{\Lambda}) = g(Y) - \frac{1}{2}\langle\tilde{\Lambda}, 
Y\zz Y-I_p\rangle + \frac{\beta}{4} ||Y\zz  Y-I_p||\fs.
\end{equation*}

The next lemma shows that the transform $Y=GX$ does not change the stationary points of problems.
\begin{lemma}\label{lemma:XYeqv}
	(i) $X^*$ is a first-order stationary point of the problem \eqref{prob:KS} if and only if $Y^*=GX^*$ is also a first-order  stationary point of the problem \eqref{prob:KS-transform-Y}.
	
	(ii) $X^*$ is a first-order stationary point of the ALM subprblem $\min_{X\in\Rnp} \cLb(X,\Lambda^*)$
	with $\Lambda^* = \mathrm{sym}(\nabla \Et(X^*)\zz X^*)$ if and only if $Y^*=GX^*$ is also a first-order stationary point of the ALM subprblem $\min_{Y\in\Rnp} \tilde{\cLb}({Y},{\tilde{\Lambda}^*})$ 	with $\tilde{\Lambda}^* = \mathrm{sym}(\nabla g(Y^*)\zz Y^*)$.

\end{lemma}
\begin{proof}
	(i) Let $Y^*=GX^*$, it can be verified that
	\begin{align*}
	(I_n-Y^*{Y^*}\zz) \nabla g(Y^*) &=   G\inv (I_n-BX^*{X^*}\zz )\nabla \Et(X^*),\\   % (I_n-GX^*{X^*}\zz G)G\inv\nabla \Et(X^*) =
	\nabla g(Y^*)\zz Y^* &= \nabla \Et({X^*})\zz {X^*}, \\\nonumber % \nabla \Et(G\inv{Y^*})\zz G\inv Y^* = 
	{Y^*}\zz  Y^* -I_p &= {X^*}\zz B{X^*}-I_p. 
	\end{align*}
	Together with \eqref{eq:FON}, we can conclude that problems \eqref{prob:KS} and \eqref{prob:KS-transform-Y} share the same first-order stationary points.

	(ii) Let $Y^*=GX^*$. Similarly, it can be verified that 
	\begin{align*}
	{\tilde{\Lambda}^*} = {\mathrm{sym}(\nabla g(Y^*)\zz Y^*)} &= {\mathrm{sym}(\nabla \Et({X^*})\zz {X^*})} = {\Lambda}^*,\\
	\nabla_Y \tilde{\cLb}({Y^*},{\tilde{\Lambda}^*}) &= G\inv \nabla_X \cLb(X^*,{\Lambda^*}).
	\end{align*}
	These equalities lead to the desired equivalence.
\end{proof}

In view of \cref{lemma:XYeqv} and let $Y=GX$, the algorithm for problem \eqref{prob:KS} can be translated into an adaptation for \eqref{prob:KS-transform-Y}. Next, we consider using PLAM to solve the problem~\eqref{prob:KS-transform-Y}. Recall that there are two major steps in the construction of PLAM: 
\begin{itemize}
	\item[1)]  For the multiplier update, we continue with the explicit update \eqref{eq:multi}, i.e.,  
	\begin{equation*}
		\tilde{\Lambda}^k=\mathrm{sym}(\nabla g(Y^k)\zz Y^k).
	\end{equation*}
	\item[2)]  We construct the subproblem with respect to $Y$,
	\begin{equation}\label{eq:PLLag-Y}
	\min\limits_{Y\in\Rnp} \, %\tilde{\cL}_\beta(X)=
	\jkh{\nabla_Y \tilde{\cLb}(Y,\tilde{\Lambda}^k), Y-Y^k}_B %\tr(\nabla_X \cLb(X^k,\Lambda^k)\zz (X-X^k))
	+ \frac{\eta_k}{2}
	\norm{Y-Y^k}\fs.
	\end{equation}
	where the inner product is defined as $\jkh{Y,\bar{Y}}_B:=\tr(Y\zz B\bar{Y})$. 
\end{itemize}

Indeed, this subproblem has the closed-form solution
\begin{align}  
Y^{k+1} &= Y^k - \frac{1}{\eta_k} B\,\nabla_Y \tilde{\cLb}(Y,\tilde{\Lambda}^k) \nonumber\\
&= Y^k - \frac{1}{\eta_k}B\left(\nabla g(Y^k) -  Y^k \varPsi(\nabla g(Y^k)\zz Y^k)+ \beta Y^k({Y^k}\zz Y^k - I_p)\right). \label{eq:PLAM-main-Y} 
\end{align}
Using $Y=GX$ and the expression of $\tilde{\Lambda}^k$, it follows that the $X$-update \eqref{eq:PLAM-main}  can be exactly recovered from~\eqref{eq:PLAM-main-Y}. In other words, the algorithm PLAM for the $X$-problem \eqref{prob:KS} is proved to be equivalent to its adaptation for the $Y$-problem \eqref{prob:KS-transform-Y}. Whereas the proximal linearized approximation in \eqref{eq:PLLag-Y} differs from what we used in \cite{gao2018parallelizable}, the sketch of the convergence analysis is nearly the same. Therefore, the convergence results for PLAM can be accordingly migrated from \cite{gao2018parallelizable}.

Finally, we present the global convergence of PLAM without proofs. Interested readers are referred to \cite{gao2018parallelizable} for a comprehensive understanding, such as the worst case complexity and local convergence rate. 
\begin{assumption}\label{assump:energy function}
	$\Et(X)$ is twice differentiable. 
\end{assumption}

\begin{assumption}\label{a2}
	For a given $X^0\in\Rnp$, we say it is a qualified initial guess, if 
	there exists $\us\in(0,1)$ such that 
	\begin{equation*}
	\sigma_{\min}(X^0)\geq \us,\qquad 0<||{X^0}\zz B X^0 -I_p||\ff \leq 1-\us^2.
	\end{equation*}
\end{assumption}
\begin{theorem}\label{thm:PLAM}
	Let $\{X^k\}$ be the iterate sequence 
	generated by \crefalg{alg:PLAM}  initialized from $X^0$ satisfying \cref{assump:energy function} and \cref{a2}. Suppose that
	the parameters $\beta$ and $\eta_k$ ($k=1,\dots$) are sufficiently large, and
	in particular, the sequence $\{\eta_k\}$ is upper bounded.
	Then the sequence $\{X^k\}$ has at least one cluster point, and any which is a first-order stationary point of problem \eqref{prob:KS}.
\end{theorem}

\subsection{An upgraded version of PLAM}
According to the numerical reports in \cite{gao2018parallelizable}, the plain PLAM performs well in most problems, whereas its behavior  is sensitive to the parameters $\beta$ and $\eta_k$. In practice, it is  always troublesome %not convenient 
to tune these parameters as
%such that 
PLAM performs identically on different problems. Even worse, we cannot guarantee the boundedness of iterate sequences without restrictions on parameters. 

Consequently, \cite{gao2018parallelizable} suggests a column-wise block minimization for PLAM to overcome these limitations. In light of its motivation, we similarly impose the redundant column-wise constraints on the subproblem \eqref{eq:PLLag},
 and obtain the following subproblem.
	\begin{equation}\label{eq:PCLag}
		\begin{array}{rcl}
			\min\limits_{X\in\Rnp} && %\tilde{\cL}_\beta(X)=
			\jkh{\nabla_X \cLb(X^k,\Lambda^k), X-X^k} %\tr(\nabla_X \cLb(X^k,\Lambda^k)\zz (X-X^k))
			+ \frac{\eta_k}{2}
			\norm{X-X^k}\fs,\\
			\st && \Diag (X\zz B X) =I.
		\end{array}
	\end{equation}	
Notice that the %objective function of \eqref{eq:PLLag} 
 subproblem \eqref{eq:PCLag}
is column-wisely separable. Thus, for the $i$-th column  ($i=1,\dots,p$), we can construct a subproblem with an extra constraint as follows,
\begin{equation}\label{eq:PLLagSph}
	\begin{array}{rcl}
		\min\limits_{x\in\Rn} && %\tilde{\cL}^{(i)}_\beta(x)= 
		\nabla_{X_i} \cLb(X^k,\Lambda^k)\zz (x-X_i^k) + \frac{\eta_k}{2}
		||x-X_i^k||_2^2,\\
		\st && x\zz Bx = 1,
	\end{array}
\end{equation}
where $X_i$ denotes the $i$-th column of $X$. The redundant constraint is for restricting the iterate sequence to a compact set and hence make it bounded. The subproblem \eqref{eq:PLLagSph} has the closed-form solution
\begin{equation}\label{eq:PCAL-main}
	X_i^{k+1} = \frac{X_i^k -\frac{1}{\eta_k}\nabla_{X_i} \cLb(X^k,\Lambda^k) }{{\norm{X_i^k -\frac{1}{\eta_k}\nabla_{X_i} \cLb(X^k,\Lambda^k)}_B}},
\end{equation}	
where $\norm{x}_B:=\sqrt{x\zz Bx}$ is a norm for any symmetric positive definite matrix $B$. Accordingly, the Lagrangian  multipliers of $X^k$ can be developed based on the new  subproblem~\eqref{eq:PCLag}. In view of these formulations, an upgraded version of PLAM is listed in \crefalg{alg:PCAL} called PCAL.%the new algorithm is listed in \crefalg{alg:PCAL} called PCAL. It is an upgraded version of PLAM.

\begin{algorithm2e}[ht]
	\caption{Parallelizable Column-wise Block Minimization for PLAM (PCAL)}
	\label{alg:PCAL}
	\SetKwInOut{Input}{input}\SetKwInOut{Output}{output}
	\SetKwComment{Comment}{}{}
	%\lines numbered
	\BlankLine %\dontprintsemicolon
	\textbf{Input:} triangulation with $n\in\mathbb{N}$ and $B\in\mathbb{SR}^{n\times n}$; tolerance $\epsilon>0$; initial guess $X^0\in\stiefel^B$; Set $k:=0$.\\
	\While{$\norm{(I_n-BX^k{X^k}\zz)H(X^k)X^k}\ff + \norm{{X^k}\zz BX^k-I}\ff > \epsilon$}
	{
		Compute the  Hamiltonian $H(X^k)$ by \eqref{eq:hamiltonian}.
		
		Compute the Lagrangian multipliers by
		\begin{equation}\label{eq:lambda2}
		\Lambda^k :=  {X^k}\zz H(X^k)\zz X^k + \Theta\left({X^k}\zz 
		\nabla_X L_\beta (X^k, {X^k}\zz H(X^k)\zz X^k)
		\right).
		\end{equation}
		
		\For{$i = 1,\dots,p$}{ 
			Update $X_i^{k+1}$ by \eqref{eq:PCAL-main}.
		}
		
		Update $X^{k+1}=[X_1^{k+1},\dots,X_p^{k+1}]$.
		
		Update the parameters $\eta_k$ and $\beta$; Set $k:=k+1$.
	}
	\textbf{Output:} $X^k$.
\end{algorithm2e}

Note that the update \eqref{eq:lambda2} for Lagrangian multipliers in PCAL is different from \eqref{eq:multi} in PLAM. When the redundant constraints, $\norm{X_i}_B=1$  ($i=1,\dots,p$), are imposed, the corresponding optimality condition changes simultaneously. Specifically, the problem \eqref{prob:KS} with redundant constraints has the first-order optimality condition as follows,
\begin{equation}\label{eq:kkt-PCAL}
\left\{
\begin{array}{l}
\nabla \Et(X) = BX\Lambda+BXD,\\ %\,\, \mbox{with~} \Lambda =\nabla \Et(X)\zz X;
X\zz BX = I_p.
\end{array}
\right.
\end{equation}
The matrix $D\in\R^{p\times p}$ is diagonal and denotes the multipliers for extra constraints.  Following a similar derivation of \eqref{eq:kkt}, it can be verified  that $\Lambda$ in \eqref{eq:kkt-PCAL} achieves the closed-form expression \eqref{eq:lambda2} at any first-order stationary point. Notice that the main calculation of PCAL is a sequence of gradient-descent step with normalization. These for-loop computations are independent and hence can be executed in a parallel fashion. To sum up, the upgraded version of PLAM still enjoys the benefit of parallel computing. 
%This is not only the parallel computing at BLAS level but also the parallel computing at algorithmic level. 

%\subsection{Preconditioning}\label{sec:precondition}
In scientific computing, preconditioning is typically used to accelerate iterative algorithms. In~\cite{bao2012h}, a preconditioner for the eigenvalue problem of SCF iteration  has been proposed. It has the form of $T=\frac{1}{2}L-\lambda B$, where $\frac{1}{2}L$ is the discretized kinetic operator defined in \eqref{prob:KS} and $\lambda$ is an approximated eigenvalue. Since $\frac{1}{2}L$ dominates the Hamiltonian, this preconditioner usually performs well in practical calculations. In view of the optimality condition \eqref{eq:kkt}, the update of Lagrangian multipliers~\eqref{eq:multi} in PLAM can be viewed as the approximation of the eigenvalues. Thus, we choose $\Lambda^k_{ii}=\dkh{{X^k}\zz H(X^k) X^k}_{ii}$ to construct a preconditioner for the proposed algorithm:
\begin{equation}\label{eq:precondition-T} 
T^k_{(i)}=\left\{\begin{array}{lc}
\frac{1}{2}L-\Lambda^k_{ii}B, & \mbox{if}~ \Lambda^k_{ii} < 0,\\
I, & \mbox{otherwise},
\end{array}\right. \quad \mbox{for~} i=1,\dots,p.
\end{equation}
Consequently, the one step gradient-descent update \eqref{eq:PLAM-main} in PLAM is preconditioned as
\begin{equation*}\label{eq:precondition-X-update}
X_i^{k+1} = X_i^k - \frac{1}{\eta_k} \dkh{T^k_{(i)}}\inv  \nabla_{X_i}\cLb(X^k,{\Lambda^k}),\quad \mbox{for~} i=1,\dots,p,
\end{equation*}
where the preconditioned gradient can be assembled by solving 
	$p$ linear systems.
Note that PCAL is compatible with this type of preconditioning providing that $\Lambda^k$ is selected from~\eqref{eq:lambda2}. The parallelizable structure of PLAM and PCAL is still maintained as the preconditioning	is conducted column-wisely. A test in \cref{fi:vs} verifies the effectiveness of the preconditioner \eqref{eq:precondition-T} for both algorithms, where substationarity is computed by Line~\ref{al:stepW} in~\crefalg{alg:PLAM}.

\begin{figure}[htbp]
	\centering
	\subfigure[PLAM for He, $\beta=15$]
	{\includegraphics[width=.45\textwidth]{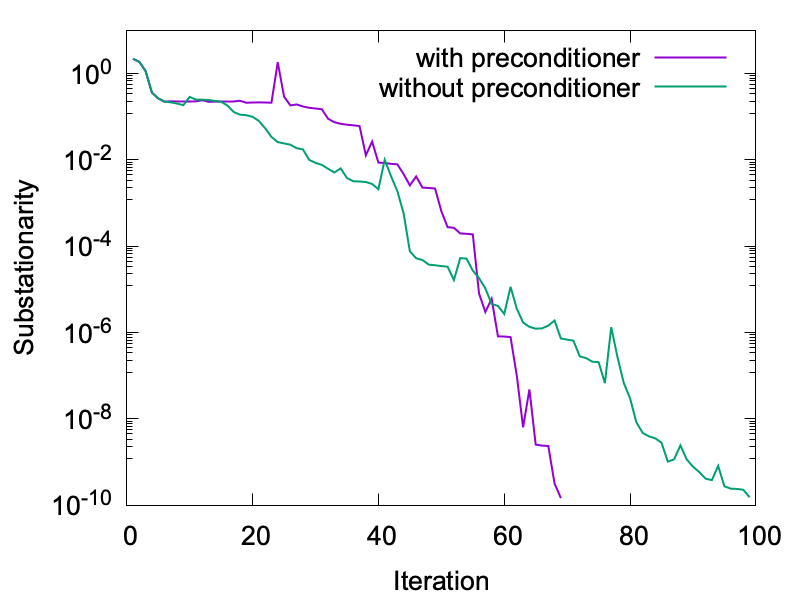}}
	\subfigure[PCAL for He, $\beta=1$]
	{\includegraphics[width=.45\textwidth]{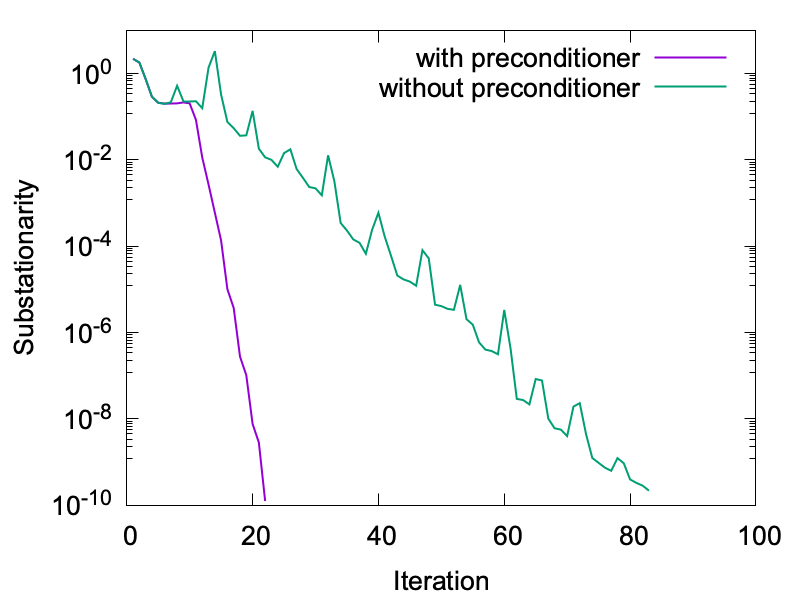}} \qquad
        \caption{The performance of the preconditioner
          \eqref{eq:precondition-T} for a helium (He) atom example
          with $n=1606, p=1$. \label{fi:vs}}
\end{figure}

\section{Implementation Details}\label{sec:imp}
In this section, we introduce the implementation details of the framework (\cref{fi:frameKS}) in solving the ground state. The quantum systems examined in this paper are introduced. In addition, several numerical issues in the simulations are discussed. In view of \cref{fi:vs}, we observe that PCAL behaves more efficient and robust than PLAM, and thus we focus on PCAL in the following tests.

All the simulations are performed on a workstation with two Intel(R) Xeon(R) Processors Silver 4110 (at 2.10GHz$\times 8$, 12M Cache) and 384GB of RAM, and the total number of cores is 16. The software is the C++ library AFEABIC \cite{bao2012h} under Ubuntu 18.10.

\subsection{Testing problems}
A number of atom and molecules are simulated to illustrate the effectiveness and high scalability of the presented algorithm. It is noted that in practical simulations, $p$ is regarded as the number of orbitals and each orbital is occupied by two electrons.  
The scale of testing systems, i.e., $p$, is ranging
from 1 to 1152. In the formulation
of problem \eqref{prob:KS-continuous}, the exchange-correlation potential $V_{\mathrm{xc}}$ and
exchange-correlation potential energy $\epsilon_{\mathrm{xc}}$ per
particle are obtained from the package Libxc
\cite{marques2012libxc}. The model equations for the various numerical
examples are only different in the external potential term
$V_{\mathrm{ext}}$
and precisely in the charge numbers and positions of the nuclei. The
charge of a certain nucleus used in numerical experiments is
listed in \cref{tab:charge}. The nuclei positions for small molecules are
obtained from the calculated geometry part in CCCBDB
\cite{nist2016cccbdb} and for carbon nanotubes are from \cite{frey2011nanotube}. In summary, the following electronic structures He ($1$), LiH (2), CH$_4$ (5), H$_{2}$O~(5), BF$_3$~(16), C$_6$H$_6$ (21), C$_{12}$H$_{10}$N$_2$~(48),
C$_{60}$ (180), and carbon nanotubes C$_{96}$ (288), C$_{192}$~(576) and C$_{384}$~(1152)
are tested, where the number in the bracket
stands for the number of orbitals $p$ in the associated system.

\begin{table}[htbp]
	\small
	\centering
	\caption{Charge number $Z_j$ of the nucleus. \label{tab:charge}}
	\begin{tabular}{ccccccccc}
		\toprule
		\multicolumn{1}{c}{}&H&He&Li&B&C&N&O&F\\ \midrule
		$Z_j$ &1&2 &3 &5&6&7&8&9 \\\bottomrule
	\end{tabular}
\end{table} 

In practice, we evaluate the values for substationarity, feasibility violation and the total energy of each example during the simulations. Specifically, $kkt=\norm{H(X)X-BX\Lambda}\ff$, $fea= \norm{X\zz B X-I}\ff$, and the total energy $\Et$ is computed from \eqref{eq:totalEnergy}. When the summation of $kkt$ and $fea$ is small enough, i.e., the following stopping criterion
\begin{equation*}
\frac{kkt+fea}{kkt_0} < tol
\end{equation*}
is satisfied, we terminate the algorithm. Here, $kkt_0$ is the initial  substationarity and $tol$ denotes the tolerance and is chosen to be
$1.0\times10^{-8}$ in our simulations.

\subsection{Pre-processing and discretization: mesh and initial guess generation}\label{sec:radialmesh}
Once we determine the computational domain, a space discretization is generated for the ground state calculation. To resolve the singularities in the external potential term, a non-uniform mesh for the partition of the computational domain is introduced to obtain high accuracy with least effort. Specifically, a global mesh size function based on the external potential is adopted to generate the
nonuniform mesh \cite{kuang2019on}. Within the linear finite element framework, to capture the $1/r$ decay in the external potential, the mesh size function locally behaves as $r^{6/5}$ for small $r$ can be derived, where $r$ represents the distance to the nucleus. 
%This kind of mesh size function is shown to be optimal for capturing the $1/r$ decay with linear interpolation.
% \cite{kuang2019on}. %suryanarayana2010non
%As a result, we could prepare the nonuniform mesh before the solution to the Kohn–Sham equation.
% In this work, 
Then we can construct the mesh size function $h(\bm{r})$ at the discretized point $\bm{r}$ as in \cite{kuang2019on}:
\begin{equation}\label{eq:meshsizefunction}
h(\bm{r})=\min \left\{\gamma_1 Z_{1}^{-\frac{2}{5}} r_{1}^{\frac{6}{5}}, \cdots, \gamma_1 Z_{M}^{-\frac{2}{5}} r_{M}^{\frac{6}{5}}, \gamma_2\right\},
\end{equation}
where $r_j = |\bm{r}-\bm{R}_j|$ represents the distance to the $j$-th
nucleus, $\gamma_1$ controls the resolution of the mesh, and
$\gamma_2$ is the largest allowed mesh size. Note that
\eqref{eq:meshsizefunction} implies that the closer to the nucleus,
the smaller the mesh size, i.e., the denser the mesh grid, which is as
desired. Moreover, the distribution of the mesh grid around the
nucleus with larger charge is also denser than that around the nucleus
with a smaller charge. This can be verified from
\cref{fi:H2OGmsh} which shows an radial mesh example for the water
molecule (H$_2$O).

\begin{figure}[htbp]
	\includegraphics[width=0.3\textwidth]{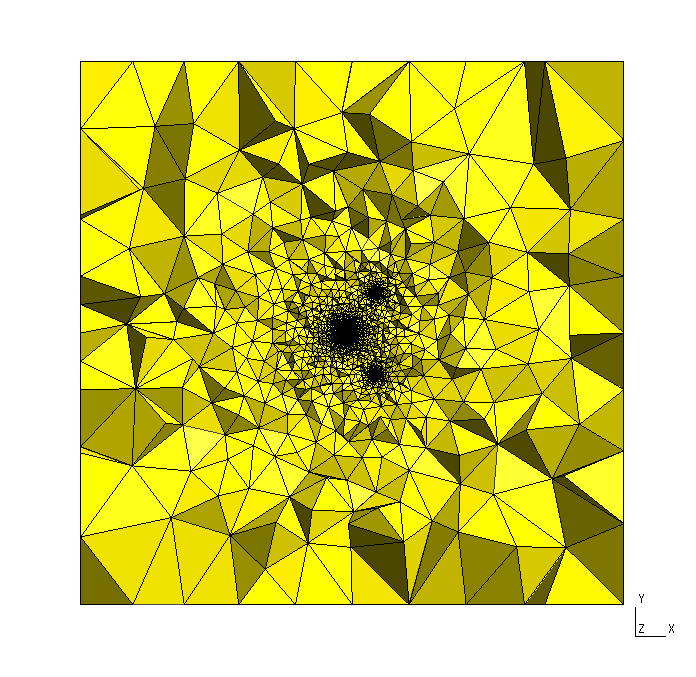}\quad
	\includegraphics[width=0.3\textwidth]{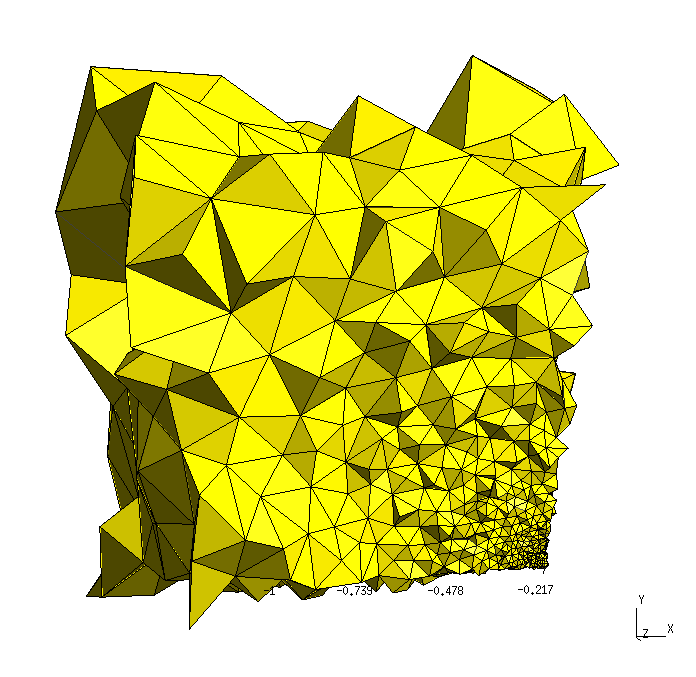}\quad
	\includegraphics[width=0.3\textwidth]{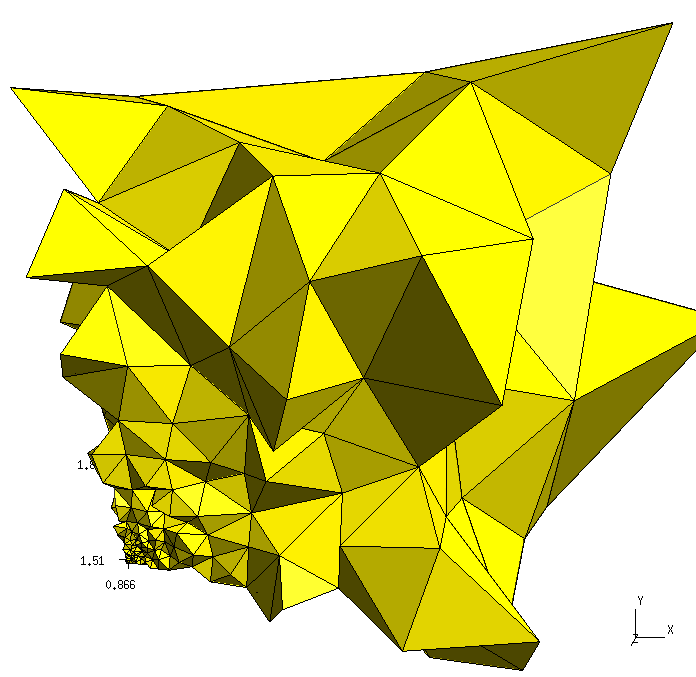}
	\caption{Left: the three dimensional mesh for molecule H$_2$O
		using the mesh size function \eqref{eq:meshsizefunction} with $\gamma_1=0.15, \gamma_2=8$. Middle: the mesh around the oxygen
		nucleus $(-0.217,0,0)$ in X-Y plane
		$[-1.217,-0.217] \times [0,1]$, on which the element shapes are
		kept. Right: the mesh around the hydrogen nucleus
		$(0.866,1.509,0)$ in X-Y plane
		$[0.866,1.866] \times [1.509,2.509]$. Generated by the software
		Gmsh v3.0.6 \cite{geuzaine2009gmsh}. \label{fi:H2OGmsh}}
\end{figure}

In the following comparison, we choose a same randomly generated initial guess, $X^0\in\Rnp$ satisfying ${X^0} \zz B{X^0}=I$, for different methods. Given a random matrix $V\in\Rnp$ from the pseudo-random number generator, $X^0$ is generated by the Cholesky-based Gram--Schmidt technique \cite{genovese2008daubechies}, i.e., $V=X^0 R$, where $R\in\R^{p\times p}$ is an upper triangular matrix. However, it is known that the SCF method may suffer a lot from divergence. For the sake of fairness, we stabilize SCF  by improving the random initial guess with the imaginary time propagation (ITP) method \cite{kuang2019on} only when it diverges. By contrast, the numerical experiments in \cref{sec:num} show that our algorithm behaves robust regardless of different initial guesses.

\subsection{Post-processing: eigenvalue evaluation in the last step\label{subsection:post-procedure}}
In view of the presented infeasible methods, it is sufficient to output results such that  $X$ satisfies the orthogonality constraint {$X\zz BX=I$}. However, if we want to extract the desired wavefunctions	from the eigenvectors of the generalized eigenvalue problem \eqref{eq:KS-gep} or the other physical quantities based on the eigenvalues, we need to introduce a post-processing. This is due to the fact that $X$ only provides an orthogonal basis of the desired eigenspace rather than the eigenvectors.

This can be implemented by solving a small $p\times p$ eigenvalue problem, $(X^{\top}HX) \tilde{X} = \lambda \tilde{X}$, with the Rayleigh-Ritz procedure to get the eigenvalues $\lambda_i, i=1,\dots,p$ and updating $X$ as $X = X\tilde{X}$ to get the wavefunctions. Note that this procedure is called only for once in the algorithm and it is of size $p\times p$. Consequently, its computational cost can be ignored compared to solving the optimization problem.

To verify the effectiveness of the post-procedure, we compute the eigenvalues of the Kohn--Sham equation of CH$_4$ system on the radial mesh with $n=100127, p=5$ for SCF and PCAL. 
%For the PCAL method without post-procedure,
%we choose the diagonal elements of the Lagrangian multiplier $\Lambda$
%as the eigenvalues $\Xi$. We also compute the eigenvalues by SCF method for
%a comparison. 
The computational domain for this example is set as $[-20,20]^3$, and the results are listed in \cref{tab:post-eig}.  When the post-procedure is imposed, the eigenvalues from PCAL are well ordered and agree with eigenvalues from SCF. Moreover, it verifies that the post-procedure does not affect the energy value. In the practical simulations, the post-procedure step will be imposed as the final step of PCAL.
\begin{table}[htbp]
	\small
	\centering
	\caption{Eigenvalue and energy evaluations for example
		CH$_4$. }
	\label{tab:post-eig}
    {\resizebox{\textwidth}{!}{
	\begin{tabular}{crrrrrr}
		\toprule
		&$\lambda_0$&$\lambda_1$&$\lambda_2$&$\lambda_3$&$\lambda_4$&$E_{KS}$ \\ \midrule
		SCF &-9.75599&-0.66451&-0.38832&-0.38830&-0.38825&-40.24109\\
		PCAL + post-processing &-9.75599&-0.66451&-0.38832&-0.38830&-0.38825&-40.24109
		\\ \bottomrule

	\end{tabular}}}
\end{table}

\subsection{Choices of parameters}\label{subsec:parameters}

There are two major parameters in the algorithm PCAL. In view of \cref{fi:vs}, the penalty parameter $\beta=1$ works well for PCAL, and hence $1$ is
set as the default value of $\beta$ in PCAL. Next, we investigate the proximal parameter $\eta_k$, whose reciprocal is the stepsize for the gradient-descent step in \crefalg{alg:PCAL}. As suggested in \cite{gao2018parallelizable}, the Barzilai--Borwein (BB) strategy \cite{BB} is an efficient way to produce the stepsize, 
\begin{equation*}
{\eta_k^{\mathrm{BB}1}} := \frac{\abs{\jkh{S^{k-1},{Y^{k-1}}}}}{\jkh{S^{k-1},S^{k-1}}}, \quad
\mbox{or} \quad {\eta_k^{\mathrm{BB}2}} :=\frac{{\jkh{Y^{k-1},{Y^{k-1}}}}}{\abs{\jkh{S^{k-1},Y^{k-1}}}},
\end{equation*}
where $S^k = X^k - X^{k-1}$, $Y^k = \nabla_{X} \cLb(X^k,\Lambda^k) -  \nabla_{X} \cLb(X^{k-1},\Lambda^{k-1})$. It has other variations such as the Alternating BB strategy \cite{Dai_Fletcher_2005},
\begin{equation*}
\eta_k^{\mathrm{ABB1}} := \left\{
\begin{array}{cl}
\eta_k^{\mathrm{BB}1},& \mbox{for odd}~~k,\\
\eta_k^{\mathrm{BB}2},& \mbox{for even}~~k,
\end{array} \right. \quad
\mbox{or}\quad
\eta_k^{\mathrm{ABB2}} := \left\{
\begin{array}{cl}
\eta_k^{\mathrm{BB}2},& \mbox{for odd}~~k,\\
\eta_k^{\mathrm{BB}1},& \mbox{for even}~~k.
\end{array}
\right.
\end{equation*}
We test PCAL with four choices of the parameter $\eta_k$ on several testing
problems. The number of iterations to achieve convergence is recorded in \cref{tab:ABBcmp}. The notation ``-'' represents that the stopping criterion has not reached after $1000$ iterations. This table reveals that PCAL with $\eta_k^{\mathrm{BB2}}$ behaves robust and has the best performance on number of iterations. As a result, we choose $\eta_k^{\mathrm{BB2}}$ as our default proximal parameter in the practical simulations.

\begin{table}[htbp]
	\small
	\centering
	\caption{Number of iterations with different proximal parameters. \label{tab:ABBcmp}}
	\begin{tabular}{rrrrrr}
		\toprule
		\multicolumn{1}{c}{} & He &LiH&CH$_4$&H$_2$O& C$_6$H$_6$\\ \midrule
		BB1& 409& - &- &-&-\\ 
		BB2& 46 &54 &75 &60& 144\\ 
		ABB1&86 &90 &129 &180 &291 \\ 
		ABB2&75 &74 &90 &119 &256\\ \bottomrule     
	\end{tabular}
\end{table}

\section{Numerical Examples}\label{sec:num}
In this section, we  numerically investigate the performance and parallel
efficiency of the algorithm PCAL in all-electron calculations under
the presented framework. 

We test the classical SCF method and MOptQR for comparisons. MOptQR is a manifold-based optimization method which can be applied to the KSDFT \cite{zhang2014gradient}. All these methods are able to fulfill the solving part in the framework described in \cref{fi:frameKS}. They are different in the main iteration: SCF solves a linear eigenvalue problem; PCAL produces a column-wise gradient-descent update; MOptQR searches along the Riemannian antigradient and projects the step onto the manifold by QR factorization. We choose the locally optimized block preconditioned conjugate gradient (LOBPCG) method \cite{knyazev2001toward} as the linear eigenvalue solver in SCF. Moreover, a simple mixing scheme is adopted for SCF, namely,  $\rho^{(k+1)} = \alpha \tilde{\rho}^{(k+1)} + (1-\alpha)\rho^{(k)}$  where $\tilde{\rho}^{(k+1)}$ is the electron density obtained from solving the $k$-th step eigenvalue problem and $\alpha$ is the mixing parameter. In both SCF method and MOptQR, the orthogonalization process is implemented by the Cholesky-based Gram-Schmidt technique \cite{genovese2008daubechies}, which is shown to be more efficient than commonly-used Gram-Schmidt procedures.

In the serial setting, the leading order of computational costs is $\mathcal{O}(np^2)$ among these methods. The reason is that BLAS3 operations, such as $X^{\top}(BX)$, dominate the computing. While the function evaluation does not have a crucial impact on the cost due to the sparsity of discretized Hamiltonian $H$ and mass matrix $B$. In the parallel setting, the total computation cost is divided into parallel and non-parallel  parts. The above-mentioned leading cost $\mathcal{O}(np^2)$ (BLAS3) belongs to the parallel part. However, the orthogonalization process in SCF and MOptQR whose complexity is $\mathcal{O}(p^3)$ cannot be efficiently parallelized. When $p$ is large, this cost is unaffordable in all-electron calculations. Conversely, PCAL is orthogonalization-free and completely consists of BLAS3 operations, and thus benefits a lot from parallel computing. These claims can be verified in the following experiments.

\subsection{Ground state calculations}\label{sec:numerical:1}
In this subsection, we test  PCAL with SCF and MOptQR in  all electron calculations of a list of atom and molecules under serial setting. For all the systems, the computational domain is set to be $[-20,20]^3$. The mesh size function~\eqref{eq:meshsizefunction} is applied to generate the nonuniform mesh for each example. Note that the parameters in~\eqref{eq:meshsizefunction}  are chosen as $\gamma_1=0.15, \gamma_2=8$ for C$_6$H$_6$ and C$_{12}$H$_{10}$N$_2$,  and $\gamma_1=0.125, \gamma_2=8$ for the others. The preconditioner \eqref{eq:precondition-T} is used in all the methods. For the system C$_{12}$H$_{10}$N$_2$, the initial guess is generated by ITP (see \cref{sec:radialmesh}), and we choose the mixing parameter $\alpha=0.15$ for SCF to make it converge. For the other systems, we choose a random initial guess and the mixing parameter $\alpha=0.3$.

The detailed numerical results are listed in \cref{tab:KS2}, \cref{fi:vsconvergence} and \cref{fi:pcal}. We observe from \cref{tab:KS2} that: 1)~the total energy $\Et$ obtained by PCAL agrees with SCF and MOptQR; 2) PCAL behaves more efficient than SCF and MOptQR in terms of the running time ``CPU(s)"; 3) the number of iterations ``$N_\mathrm{iter}$" in PCAL is less than MOptQR. Note that the iteration numbers of SCF are always the smallest but conversely	the CPU time. This is due to that the inner iterations, i.e., solving the linear eigenvalue problem, are required in each SCF iteration. The efficiency of PCAL can be also observed in \cref{fi:vsconvergence} for the example C$_6$H$_6$, from which we find that PCAL takes the least CPU time to converge at a given accuracy. In addition, the convergence results for	PCAL are demonstrated in \cref{fi:pcal}.  The first column displays the isosurface of	the electron density, the last three columns present the convergence	history of energy,  substationarity and feasibility violation, respectively. We observe that the feasibility violation of PCAL gradually decreases until it converges. Note that the post-processing is not shown in this figure. In the He example, the feasibility violation is close to the machine accuracy since the normalization procedure is equivalent to the orthogonalization procedure in the case of $p=1$.

\begin{table}[htbp]
	\centering
	\caption{The results in Kohn--Sham total energy minimization\label{tab:KS2}} 
	{\resizebox{\textwidth}{!}{
	\begin{tabular}{c|ccrcr|ccrcr}
		\toprule
		Solver& $\Et$ & $kkt$ & $N_\mathrm{iter}$& $fea$ & CPU(s) & $\Et$ & $kkt$ & $N_\mathrm{iter}$& $fea$ & CPU(s) \\ \midrule
		\multicolumn{1}{c}{} & \multicolumn{5}{c}{He, $n=34481, p=1$} & \multicolumn{5}{c}{LiH, $n=63725, p=2$}\\ \midrule
		SCF  &-2.86809&9.58$_{-9}$& 36&2.88$_{-15}$&127&-7.98190&2.15$_{-7}$&39 &2.09$_{-14}$&617\\
		MOptQR&-2.86808&3.56$_{-8}$& 36&3.55$_{-15}$&82&-7.98190&1.36$_{-7}$&70 &2.70$_{-15}$&317\\
		PCAL &-2.86808&5.99$_{-9}$& 46&1.62$_{-15}$&99&-7.98190&1.25$_{-7}$&54 &2.24$_{-15}$&264\\\midrule
		
		\multicolumn{1}{c}{} & \multicolumn{5}{c}{CH$_4$, $n=141189, p=5$} & \multicolumn{5}{c}{H$_2$O, $n=149616, p=5$}\\ \midrule
		SCF  &-40.23775&4.83$_{-7}$& 39&6.09$_{-14}$&3788&-75.83672&1.07$_{-7}$& 44&3.64$_{-14}$&4246\\
		MOptQR&-40.23775&1.24$_{-7}$& 93&2.67$_{-14}$&1721&-75.83672&1.48$_{-7}$& 74&2.45$_{-14}$&1413\\
		PCAL &-40.23775&5.66$_{-6}$& 75&1.59$_{-14}$&1283&-75.83672&1.33$_{-7}$& 60&4.82$_{-14}$&1219\\\midrule
		
		\multicolumn{1}{c}{} & \multicolumn{5}{c}{C$_6$H$_6$, $n=241939, p=21$} & \multicolumn{5}{c}{C$_{12}$H$_{10}$N$_2$, $n=522149, p=48$}\\ \midrule
		SCF  &-231.05824&5.11$_{-7}$& 41&1.73$_{-13}$&43901&-571.60648&2.57$_{-8}$&94 &1.64$_{-13}$&379482\\
		MOptQR&-231.05824&3.60$_{-7}$&269&5.14$_{-14}$&21238&-571.60648&7.30$_{-8}$&501&1.63$_{-13}$&225856\\
		PCAL &-231.05824&3.71$_{-7}$&144&7.35$_{-14}$&11013&-571.60648&5.37$_{-8}$&148&2.29$_{-13}$&89116\\
		\bottomrule
	\end{tabular}}}
\end{table} 

\begin{figure}[htbp]
	\centering
	\subfigure[Energy value]{
	\includegraphics[width=0.45\textwidth]{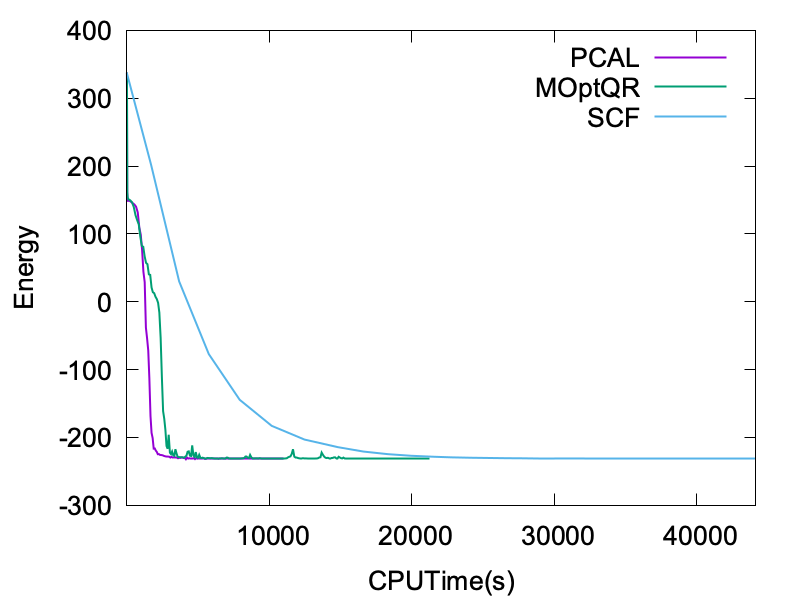}}
	\qquad
	%\subfigure[KKT violation]{
	\subfigure[Substationarity]{
	\includegraphics[width=0.45\textwidth]{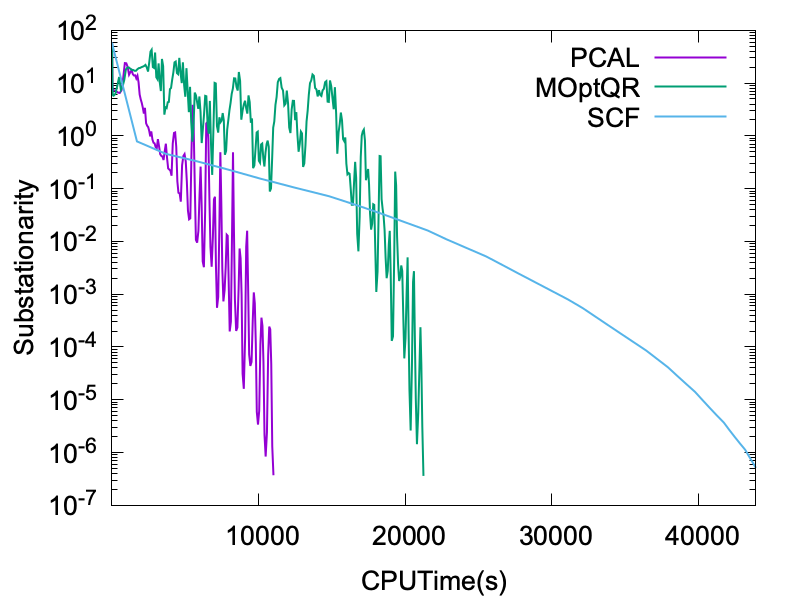}}

        \caption{A comparison with different solvers for example
		C$_6$H$_6$. \label{fi:vsconvergence}}
\end{figure}

\begin{figure}[htbp]
	\centering
	\includegraphics[width=0.24\textwidth]{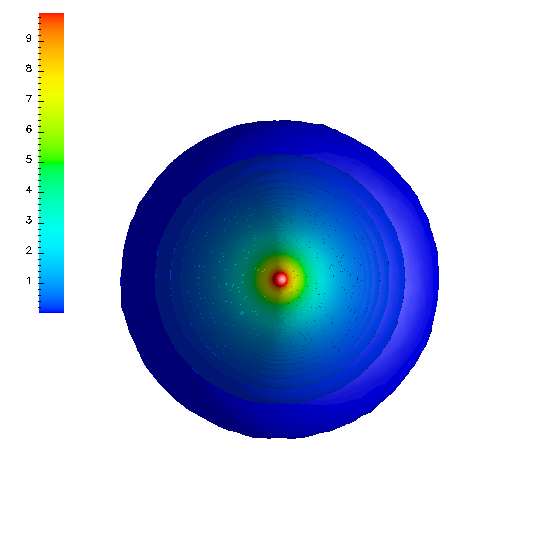}\hfill
	\includegraphics[width=0.24\textwidth]{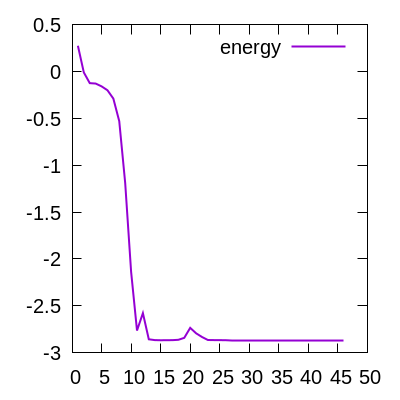}
	\includegraphics[width=0.24\textwidth]{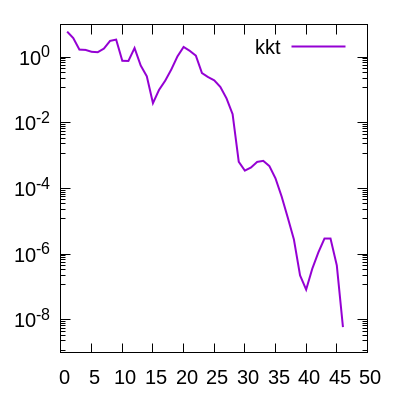}
	\includegraphics[width=0.24\textwidth]{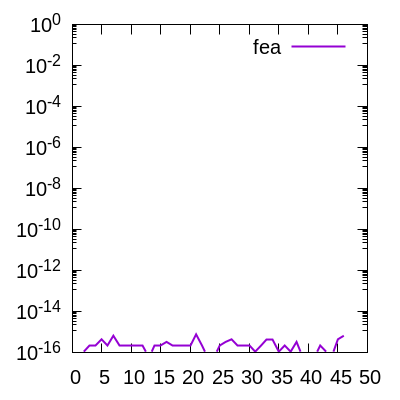}

	\includegraphics[width=0.24\textwidth]{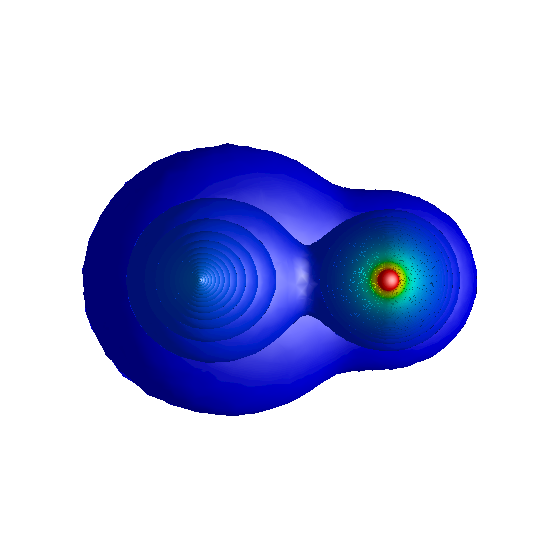}\hfill
	\includegraphics[width=0.24\textwidth]{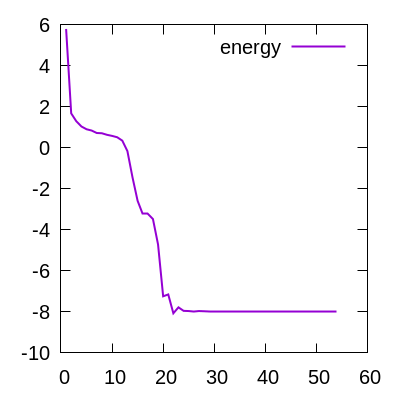}
	\includegraphics[width=0.24\textwidth]{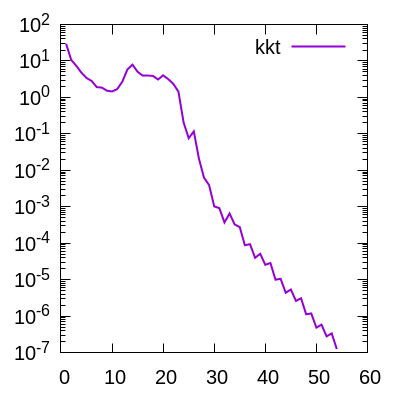}
	\includegraphics[width=0.24\textwidth]{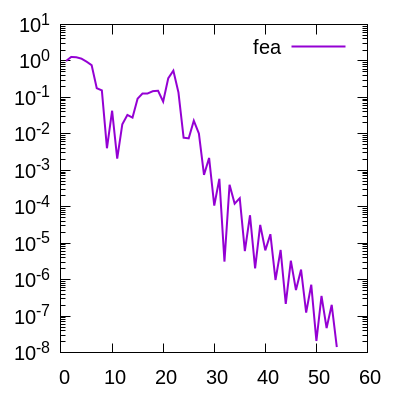}

	\includegraphics[width=0.24\textwidth]{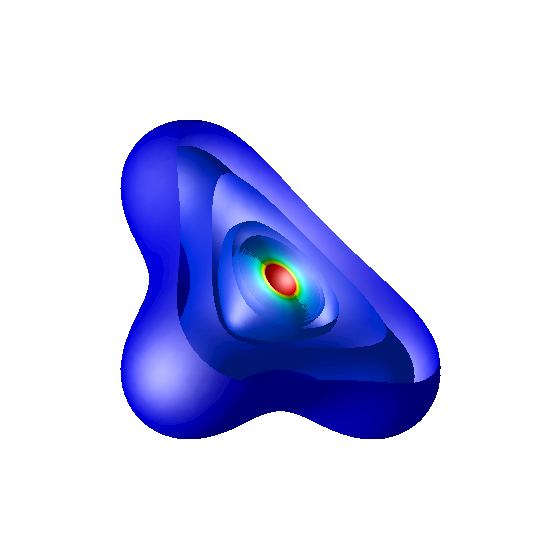}\hfill
	\includegraphics[width=0.24\textwidth]{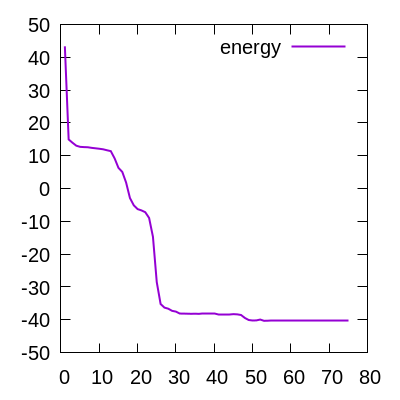}
	\includegraphics[width=0.24\textwidth]{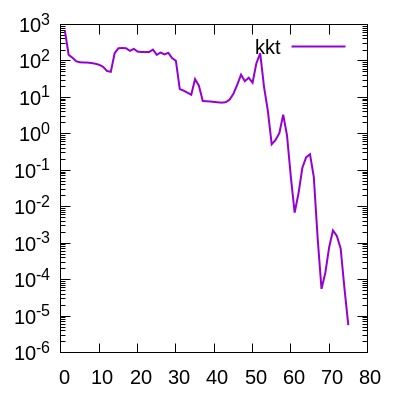}
	\includegraphics[width=0.24\textwidth]{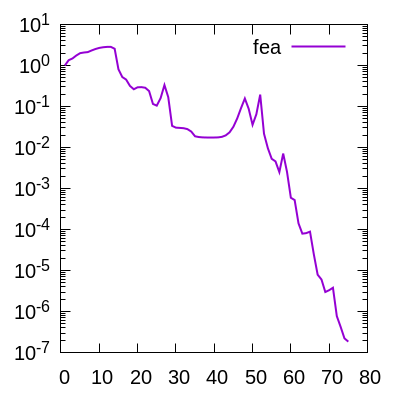}

	\includegraphics[width=0.24\textwidth]{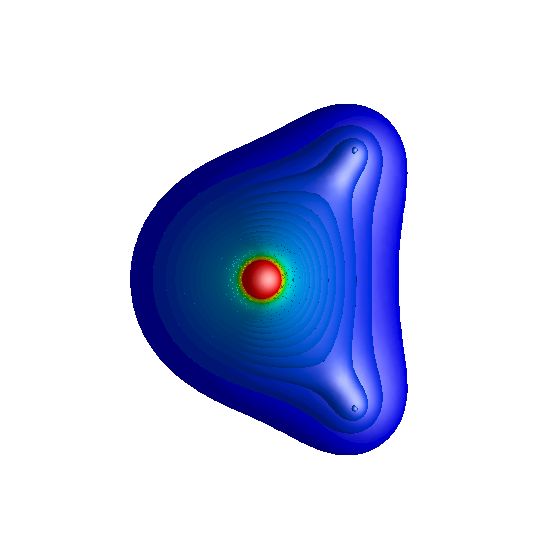}\hfill
	\includegraphics[width=0.24\textwidth]{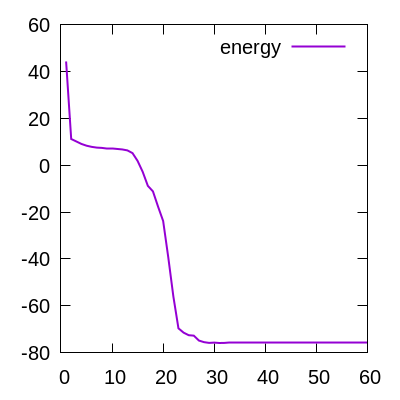}
	\includegraphics[width=0.24\textwidth]{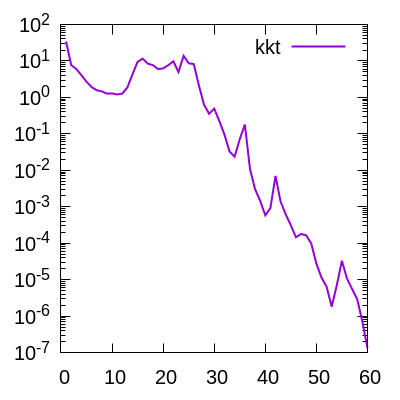}
	\includegraphics[width=0.24\textwidth]{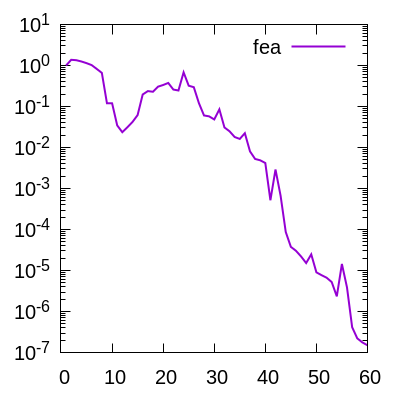}

	\includegraphics[width=0.24\textwidth]{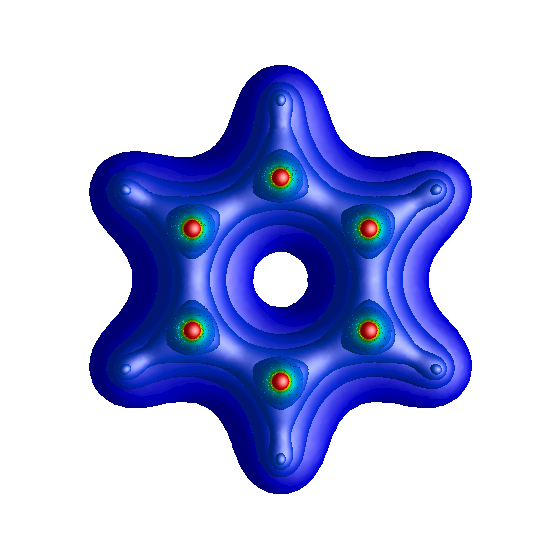}\hfill
	\includegraphics[width=0.24\textwidth]{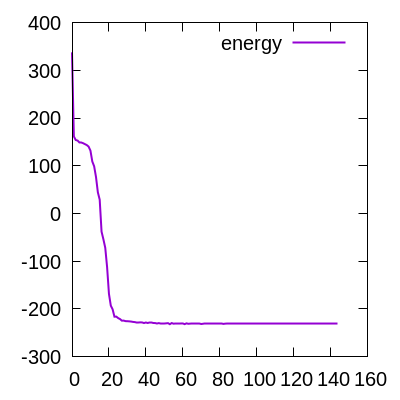}
	\includegraphics[width=0.24\textwidth]{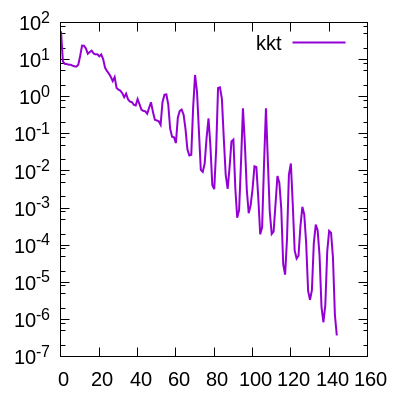}
	\includegraphics[width=0.24\textwidth]{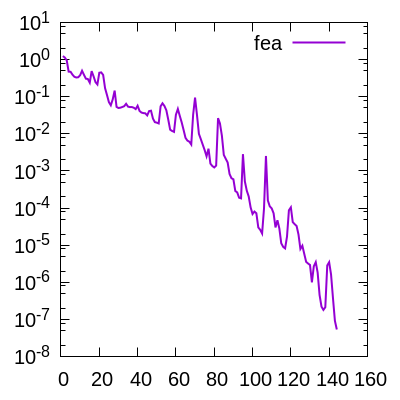}

	\includegraphics[width=0.24\textwidth]{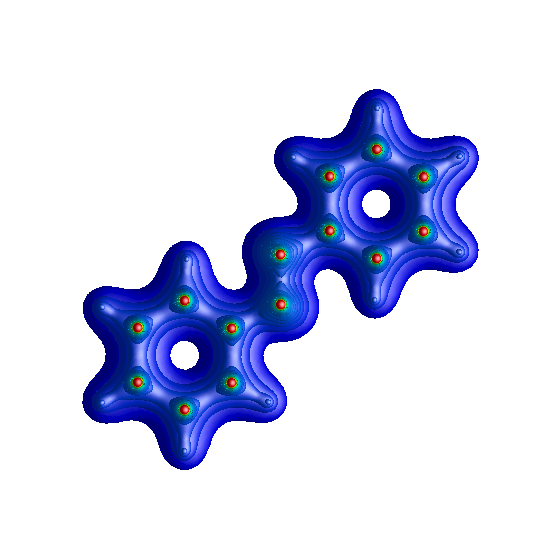}\hfill
	\includegraphics[width=0.24\textwidth]{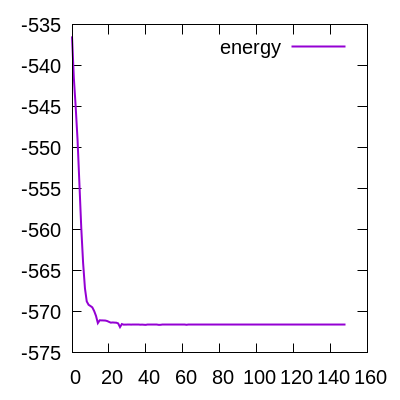}
	\includegraphics[width=0.24\textwidth]{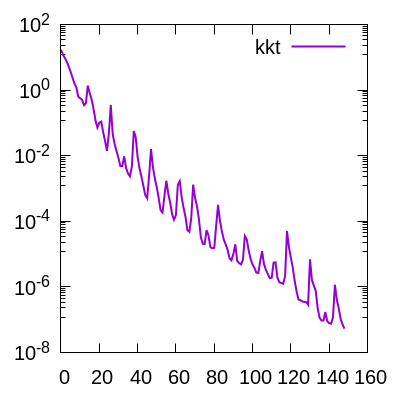}
	\includegraphics[width=0.24\textwidth]{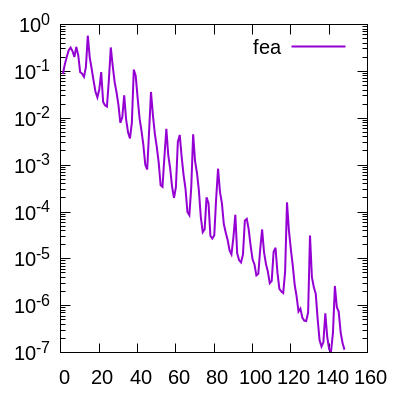}

	\caption{Convergence history of PCAL for He, LiH, CH$_4$, H$_2$O, C$_6$H$_6$,	C$_{12}$H$_{10}$N$_2$ (from top to bottom). The left column displays the isosurface of each molecule. $x$-axis for the right 3 columns	stands for the iteration step. \label{fi:pcal}}
\end{figure}

\subsection{Scalability}\label{sec:scalability}
In this subsection, we investigate the parallel efficiency of PCAL. We first test all the algorithms on a single core and record the computational proportions of parallel and non-parallel part in the total cost. We fix $n$ to be around $100000$ and choose different $p$, namely, the different molecules. By adjusting the parameters in \cref{sec:radialmesh}, we control the number of mesh grids $n$ being as close as possible to $100000$. The testing examples are  BF$_3$ (16), C$_{12}$H$_{10}$N$_2$ (48), C$_{60}$ (180), C$_{96}$ (288), and C$_{192}$ (576). The numerical results are displayed in \cref{fi:PervsP}. We observe that the parallel part of PCAL will dominate the total cost when $p$ becomes large. It is even higher than $99\%$ When $p\ge 180$. This means that PCAL is suitable for parallel computing, especially when the scale of a system is very large. Meanwhile, it can be found that the non-parallel part of SCF and MOptQR becomes large when $p$  increases. The main reason is the cubic complexity $\mathcal{O}(p^3)$ of orthogonalization process, which cannot be parallelized.

\begin{figure}[htbp] 
	\centering
	\subfigure[Parallel]{
		\includegraphics[width=0.45\textwidth]{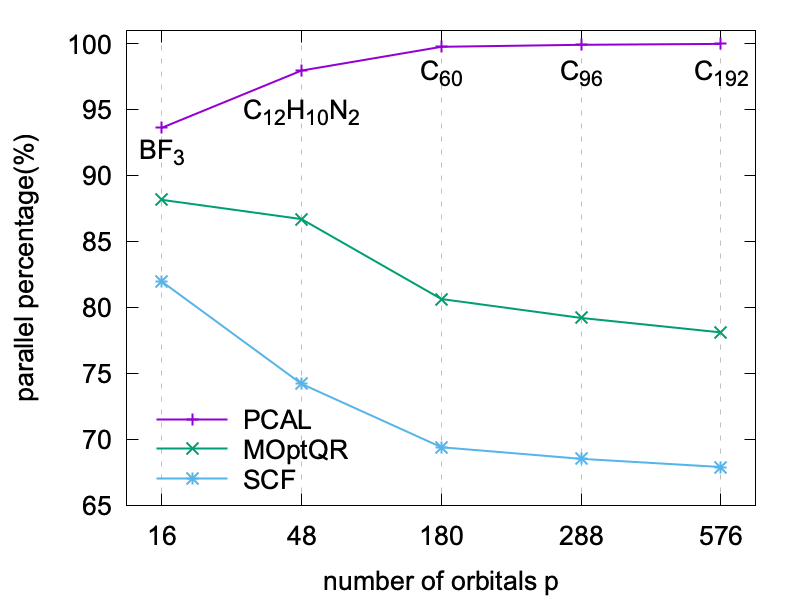}}\qquad
	\subfigure[Non-parallel]{
	\includegraphics[width=0.45\textwidth]{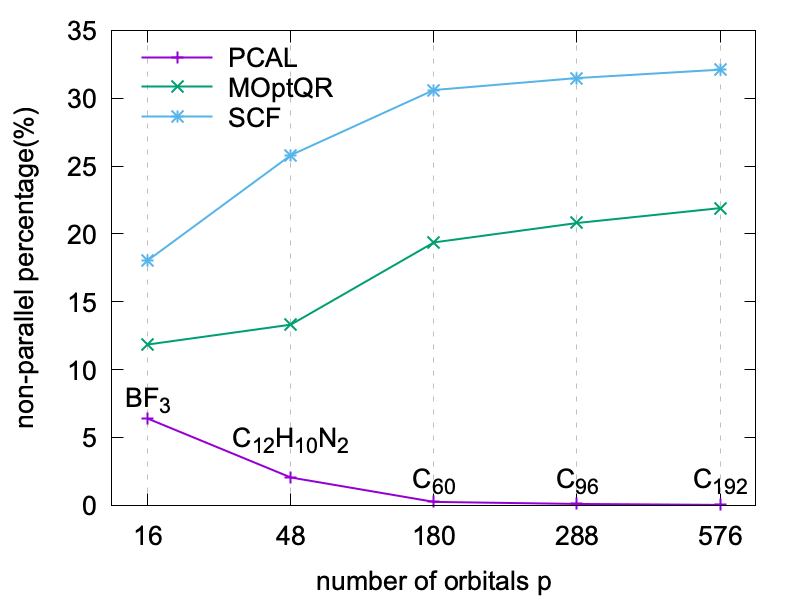}}
	\caption{Parallel proportion versus number of orbitals $p$.\label{fi:PervsP}}
\end{figure}

We next examine the scalability of PCAL in the parallel setting. The testing molecule is C$_{384}$ which has $1152$ occupied orbitals. The number of mesh grids $n$ is set to be $380233$. We run the code on different numbers of cores  $\{4, 8, 12, 16\}$. The corresponding speedup factor is defined as
\begin{equation*}
\text { speedup-factor }(m)=\frac{\text { wall-clock time for 4-core run }}{\text { wall-clock time for a } m\text {-core run }}.
\end{equation*}
The results are presented in \cref{fi:c384sf}, from which we observe that the speedup factor of PCAL is close to the ideal one, and it achieves $3.76$ for $16$ cores. However, MOptQR has the low scalability and its speedup factor increases slowly. Note that the results of SCF are not recorded since the divergent phenomenon is observed. In view of \cref{fi:PervsP}, even if we have the convergent results of SCF, it can be justifiably expected that the speedup factor of SCF will be smaller than that of MOptQR. In summary,  the orthogonalization-free algorithm PCAL shows higher scalability and great potential than SCF and MOptQR.

\begin{figure}[htbp]
  \centering
  \parbox{\textwidth}{
    \parbox{0.31\textwidth}{
      \centering
      \subfigure[Structure]
      {\includegraphics[width=0.2\textwidth]{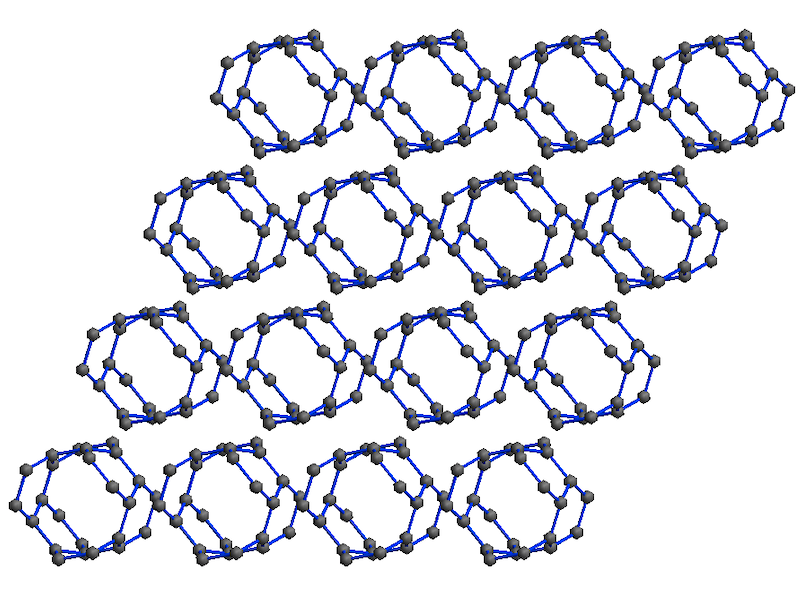}}
     % \vskip1em
      \subfigure[Isosurface]
      {\includegraphics[width=0.31\textwidth]{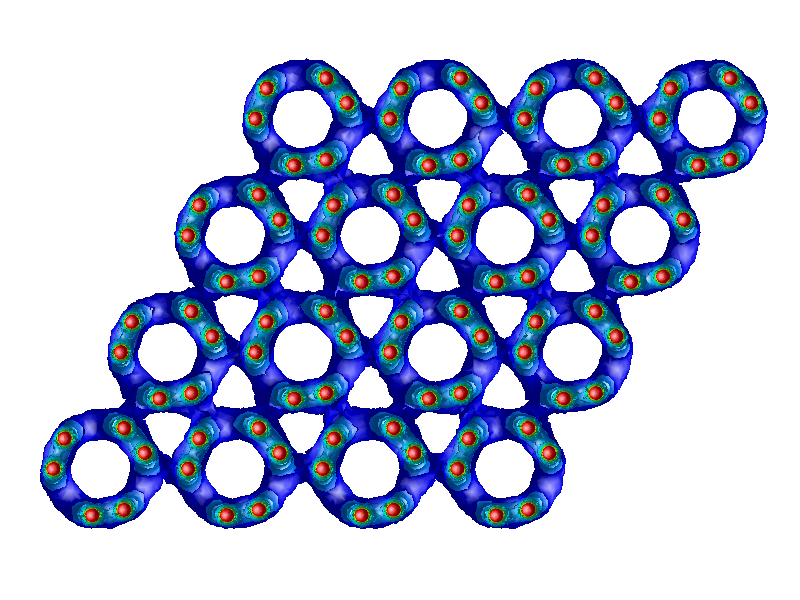}}
    }
    \hspace{9mm}
    \parbox{0.6\textwidth}{
      \subfigure[Speedup factor]
      {\includegraphics[width=\hsize]{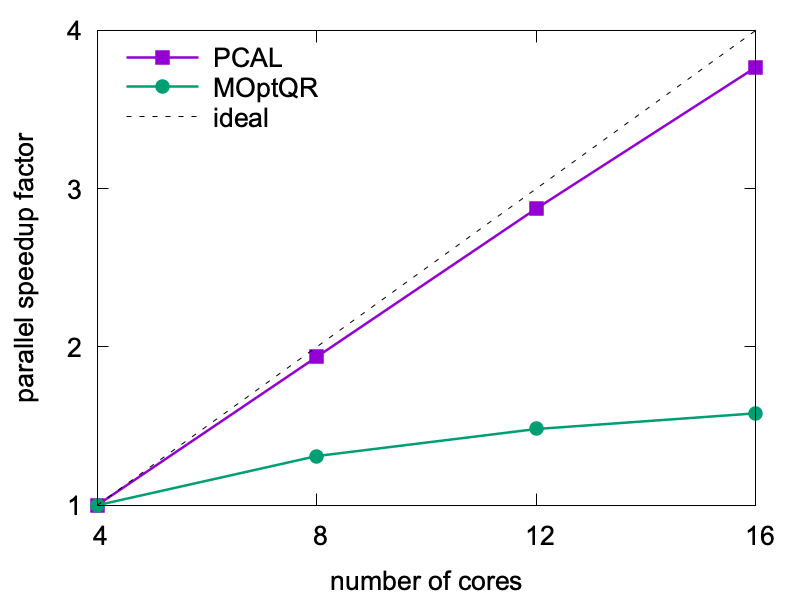}}
    }
  }
  
  \caption{Example C$_{384}$ with $n=380233$, $p=1152$.\label{fi:c384sf}}
\end{figure}

\section{Conclusion}
Based on the finite element method and PCAL algorithm, a scalable approach is proposed in this paper for the ground state solution of a given quantum system. To resolve the singularity introduced from the all-electron model, a radial mesh is generated according to the structure of the system, then the optimization problem is discretized in the associated finite element space. To avoid the efficiency bottleneck for large scale systems, i.e., the  orthogonalization of those orbitals, the original PCAL method is extended and applied in this paper for solving the discretized optimization problem. A novel preconditioner is designed in the extended PCAL method, which generally accelerates the convergence in the simulations.

Comprehensive numerical experiments are implemented for different molecules. The effectiveness of the proposed method is well demonstrated by the comparison among the proposed method, the classical SCF method, and the MOptQR method in serial computing. Meanwhile, the robustness of the proposed method is fully demonstrated by its insensitivity to the initial guess and the algorithm parameters. The feature of the proposed method on improving the efficiency by avoiding the orthogonalization procedure is displayed clearly by the huge reduction of the CPU time in the comparison to the SCF method. More importantly, the excellent scalability of the proposed method is successfully shown in an experiment on a relatively large scale electronic system.

To improve the proposed method, the $h$-adaptive mesh method will be introduced for dynamically adjusting the finite element space according to the obtained numerical solutions. Furthermore, the preconditioner introduced in the PCAL method deserves more investigation in the following study, which has a chance to effectively accelerate the convergence of the numerical method towards the ground state. The improved method will be used for the numerical simulations of the Born-Oppenheimer molecular dynamics, to show the potential on the practical applications. The results will be reported in the forthcoming paper.

%\clearpage
% -------------------------
% make bib in alphabet order with SIAM template
%\nocite{*}
\bibliographystyle{siamplain}
\bibliography{bibfile}
% -------------------------

\end{document}